\documentclass[preprint,1p,number]{elsarticle}

\usepackage{url}
\usepackage[english]{babel}
\usepackage{amsmath}
\usepackage{amsthm}
\usepackage{amsfonts}
\usepackage{amssymb}
\usepackage{natbib}
\usepackage{graphicx,enumerate,algorithm,algorithmic,multicol}

\biboptions{numbers,square,sort}

\newcommand{\eqi}{\gamma}
\newcommand{\desc}{\tau}

\newcommand{\Tr}{^\mathrm{\small T}}

\newcommand{\kerm}{\mathrm{ker}}

\newcommand{\va}{\alpha}
\newcommand{\vx}{{x}}
\newcommand{\vy}{{y}}
\newcommand{\parenth}[1]{\left({#1}\right)}

\newcommand{\Pm}{\mathrm{P}}

\newcommand{\Id}{\boldsymbol{\mathrm{I}}}
\newcommand{\Amb}{\boldsymbol{\mathrm{A}}}

\newcommand{\Fmb}{\boldsymbol{\mathrm{F}}}
\newcommand{\Hmb}{\boldsymbol{\mathrm{H}}}

\newcommand{\Phb}{\boldsymbol{\Phi}}

\newcommand{\opnorm}[1]{\left\lvert\!\left\lvert\!\left\lvert #1 \right\rvert\!\right\rvert\!\right\rvert}
\newcommand{\norm}[1]{\left\lVert #1 \right\rVert}
\newcommand{\pds}[2]{\left\langle#1,#2\right\rangle}
\newcommand{\abs}[1]{\left\lvert #1 \right\rvert}
\newcommand{\trans}{^{\mathrm{T}}}

\newcommand{\Cc}{\mathcal{C}}

\newcommand{\RR}{\mathbb{R}}
\newcommand{\Hc}{\mathcal{H}}

\newcommand{\Mc}{\mathcal{M}}

\newcommand{\Oc}{\mathcal{O}}
\newcommand{\Pc}{\mathcal{P}}

\newcommand{\Prj}{\mathcal{P}}

\DeclareMathOperator{\prox}{prox}

\DeclareMathOperator{\dom}{dom}
\DeclareMathOperator{\ri}{ri}

\DeclareMathOperator*{\argmin}{argmin}

\DeclareMathOperator*{\arginf}{arginf}

\DeclareMathOperator{\card}{card}

\renewcommand{\ge}{\geqslant}
\renewcommand{\le}{\leqslant}

\theoremstyle{plain}
\newtheorem{theorem}{\textbf{Theorem}} 

\newtheorem{lemma}[theorem]{\textbf{Lemma}}

\newtheorem{propo}[theorem]{\textbf{Proposition}}

\newtheorem{definition}[theorem]{\textbf{Definition}}

\graphicspath{{images/}}

\begin{document}

\begin{frontmatter}



\title{Deconvolution under Poisson noise using exact data fidelity and synthesis or analysis sparsity priors}


\author[greyc,cea]{F.-X. Dup\'e\corref{cor1}}
\ead{francois-xavier.dupe@cea.fr}

\author[greyc]{M. J. Fadili}
\ead{Jalal.Fadili@greyc.ensicaen.fr}

\author[cea]{J.-L. Starck}
\ead{jean-luc.starck@cea.fr}

\cortext[cor1]{Corresponding author}

\address[greyc]{GREYC UMR CNRS 6072, Universit\'e de Caen Basse-Normandie/ENSICAEN, 6 Bd Mar\'echal Juin, 14050 Caen, France}
\address[cea]{Laboratoire AIM, UMR CEA-CNRS-Paris 7, Irfu, SEDI-SAP, Service d'Astrophysique,  CEA Saclay, F-91191 GIF-Sur-YVETTE CEDEX, France}

\begin{abstract}
  In this paper, we propose a Bayesian MAP estimator for solving the deconvolution problems when the observations are
  corrupted by Poisson noise. Towards this goal, a proper data fidelity term (log-likelihood) is introduced to reflect
  the Poisson statistics of the noise. On the other hand, as a prior, the images to restore are assumed to be positive
  and sparsely represented in a dictionary of waveforms such as wavelets or curvelets. Both analysis and synthesis-type
  sparsity priors are considered. Piecing together the data fidelity and the prior terms, the deconvolution problem
  boils down to the minimization of non-smooth convex functionals (for each prior). We establish the well-posedness of
  each optimization problem, characterize the corresponding minimizers, and solve them by means of proximal splitting
  algorithms originating from the realm of non-smooth convex optimization theory. Experimental results are conducted to
  demonstrate the potential applicability of the proposed algorithms to astronomical imaging datasets.
\end{abstract}


\end{frontmatter}

\section{Introduction}
\label{sec:intro}

Deconvolution is a longstanding problem in many areas of signal and image processing (e.g. biomedical imaging
\cite{Sarder2006,Pawley2005,Bertero2009}, astronomy \cite{Starck2006,Bertero2009}, remote-sensing, to quote a few). For
example, research in astronomical image deconvolution has seen considerable work, partly triggered by the
Hubble space telescope (HST) optical aberration problem at the beginning of its mission. In biomedical imaging,
researchers are also increasingly relying on deconvolution to improve the quality of images acquired using complex
optical systems, like confocal microscopes \cite{Pawley2005}. Deconvolution may then prove crucial for exploiting images
and extracting scientific content, and is more and more involved in everyday research.

An extensive literature exist on the deconvolution problem \cite{AndrewsHunt77,Stark1987,Jansson97}. In presence of Poisson noise, several deconvolution methods have been proposed such as Tikhonov-Miller inverse filter and Richardson-Lucy (RL) algorithms; see \cite{Sarder2006,Starck2006} for comprehensive reviews. The RL has been extensively used in many applications, but as it tends to amplify the noise after a few iterations, several extension have been proposed. For example, wavelet-regularized RL algorithm has been proposed by several authors \cite{Starck94,Starck95,Bijaoui04}. The interested reader may refer to \cite{Bertero2009,Dupe2009c} for a more complete review on deconvolution with Poisson
noise.

In the context of deconvolution with either Gaussian or Poisson noise, sparsity-promoting regularization over orthobasis
or frame dictionaries has been recently proposed
\cite{Starck03,Figueiredo2003,Daubechies2004,Combettes2005,Fadili2006a,Combettes2007a,Dupe2009c,Pustelnik2009,Figueiredo2010}.

For instance, in \cite{Dupe2009c} the authors proposed to stabilize the Poisson noise using the Anscombe variance
stabilizing transform \cite{Anscombe48} to bring the problem back to deconvolution with additive white Gaussian noise,
but at the price of a non-linear data fidelity term. However, stabilization has a cost, and the performance of such an
approach degrades in low intensity regimes. On the other hand, the authors in \cite{Pustelnik2009} use directly the data
fidelity term related to Poisson noise (see Section~\ref{sec:problem-statement}) and exploit its additive structure
together with some particularities of the composition of the convolution operator with tight frames in order to solve
the problem using a proximal framework. However, their framework only applies to some special convolution kernels. Using
the augmented Lagrangian method with the alternating method of multipliers algorithm, \cite{Figueiredo2010} presented a
deconvolution algorithm with either a synthesis or an analysis prior. In fact, ADMM is nothing but the Douglas-Rachford
splitting \cite{Eckstein92} applied to the Fenchel-Rockafellar dual objective.

\paragraph{Contributions}
In this paper, we propose an image deconvolution algorithm for data blurred and contaminated by Poisson noise using
analysis and synthesis-type sparsity priors. In order to form the data fidelity term, we take the exact Poisson
likelihood.  Putting together the data fidelity and the prior terms, the deconvolution problem is formulated as the
minimization of a maximum a posteriori (MAP) objective functional involving three terms: the data fidelity term; a
non-negativity constraint (as Poisson data are positive by definition); and a regularization term, in the form of a
non-smooth penalty that enforces the sparsity of the sought after image over a {\textemdash}possibly redundant{\textemdash} dictionary of waveforms. We establish the well-posedness of our optimization problems and characterize the corresponding minimizers. We then solve them by means of proximal splitting algorithms originating from the field of non-smooth convex optimization theory. More precisely we take benefit from a generalization of Douglas-Rachford splitting to product spaces; see Section~\ref{sec:gener-minim-algor}. In order to use this proximal splitting algorithm, the proximity operator of each individual term in the objective functional is established. Experimental results on several imaging datasets are carried out to demonstrate the potential applicability of the proposed algorithms and compare them with some competitors.

\paragraph{Notation and terminology}
\label{sec:notation}

Let $\Hc$ a real Hilbert space, here a finite dimensional real vector space. We denote by $\norm{.}$ the norm associated
with the inner product $\pds{.}{.}$ in $\Hc$, and $\Id$ is the identity operator on $\Hc$. $\vx$ and $\va$ are
respectively reordered vectors of image samples and coefficients. We denote by $\ri \Cc$ the relative interior
of a convex set $\Cc$.

A real-valued function $f$ is coercive, if $\lim_{\norm{z} \to +\infty}f\parenth{z}=+\infty$. The domain of $f$ is
defined by $\dom f = \{ z\in\Hc \mid f(z) < +\infty \}$ and $f$ is proper if $\dom f \neq \emptyset$. A real-valued
function $f$ is lower semi-continuous (lsc) if $\liminf_{z \to z_0} f(z) \ge f(z_0)$. $\Gamma_0(\Hc)$ is the class of
all proper lsc convex functions from $\Hc$ to $(-\infty,+\infty]$. The subdifferential of a function $f \in
\Gamma_0(\Hc)$ at $z \in \Hc$ is the set $\partial f(z) = \left\{u \in \Hc \mid \forall y \in \Hc, f(y) \ge f(z) +
  \pds{u}{y-z}\right\}$. Its Legendre-Fenchel conjugate is $f^*$. We denote by $\imath_{\Cc}$ the indicator of the convex set $\Cc$: $ \imath_{\Cc} (z) =
  \begin{cases}
    0, & \text{if } z \in \Cc ~ ,\\
    +\infty, & \text{otherwise}.
  \end{cases}$.
See \cite{LemarechalHiriart96} for a more comprehensive account on convex analysis.

We denote by $\opnorm{\mathbf{M}}= \max_{z \neq 0} \frac{\norm{\mathbf{M}z}}{\norm{z}}$ the spectral norm of
$\mathbf{M}$.
  
\section{Sparsity prior}
\label{sec:sparse-image-repr}

Let $x \in \RR^n$ be an $\sqrt{n}\times\sqrt{n}$ image. We denote by $\Phb$ the dictionary, i.e. the $n\times L$ matrix
whose columns are the generating waveforms (called atoms) $\parenth{\varphi_i}_{1 \leq i \leq L}$ all normalized to a
unit $\ell_2$-norm. The forward (analysis) transform is then defined by a non-necessarily square matrix $\Phb\Tr \in
\mathbb{R}^{L\times n}$ with $L\ge n$. When $L > n$ the dictionary is said to be redundant or overcomplete. A dictionary
matrix $\Phb$ is said to be a frame with bounds $c_1$ and $c_2$, $0<c_1\le c_2<+\infty$, if $ c_1 \norm{x}^2 \le
\norm{\Phb\Tr x}^2 \le c_2 \norm{x}^2$~.  A frame is tight when $c_1=c_2=c$, i.e. $\Phb\Phb\Tr =c\Id$.

In a synthesis prior model, the image $x$ can be written as the superposition of the elementary atoms $\varphi_i$
according to the following linear generative model $x = \sum_{i=1}^L \va[i] \varphi_i = \Phb \va$. $x$ is said to be
sparse (resp. compressible)\footnote{With a slight abuse of terminology, in both cases we will use the term sparse.} in
$\Phb$ if it can be represented exactly (resp. approximately) as a superposition of a small fraction of the atoms in the
dictionary $\Phb$ compared to the ambient dimension $n$. In an analysis-type prior model, the transform coefficients
$\Phb\Tr x$ of the image $x$ are assumed to be sparse.

In the sequel, the dictionary $\Phb$ corresponds to a frame and will be built by taking union of one or several
transforms. For instance, for astronomical objects such as stars, the wavelet transform is a very good candidate
\cite{Mallat98}. For more anisotropic or elongated structures, ridgelets or curvelets would be a better choice
\cite{CandesDonohoCurvelets}.

%
%
%
%

\section{Problem statement}
\label{sec:problem-statement}

Consider the image formation model where an input image of $n$ pixels $\vx$ is blurred by a point spread function (PSF)
$h$ and contaminated by Poisson noise. The observed image is then a discrete collection of counts $\vy=(\vy_i)_{1 \le i
  \le n}$ which are bounded, i.e. $\vy \in \ell_{\infty}$.  Each count $y_i$ is a realization of an independent Poisson
random variable with a mean $(h \circledast \vx)_i$, where $\circledast$ is the circular convolution operator. Formally,
this writes,
\begin{equation}
  \label{eq:2}
  \vy_i \sim \Prj\parenth{(h \circledast \vx)_i}~.
\end{equation}
This formula can be rewritten in a vector form as $\vy \sim \Pc(\Hmb \vx)$, where $\Hmb$ is a circular convolution
matrix. The deconvolution problem at hand is to restore $\vx$ from the observed count image $\vy$.

A natural way to attack this problem would be to adopt a maximum a posteriori (MAP) bayesian framework with an
appropriate likelihood function {\textemdash} the distribution of the observed data $\vy$ given an original $\vx$
{\textemdash} reflecting the Poisson statistics of the noise. As a prior, the image is supposed to be economically
(sparsely) represented in a pre-chosen dictionary $\Phb$ as measured by a sparsity-promoting penalty $\Psi$ supposed
throughout to be proper, lsc and convex but non-smooth, e.g. the $\ell_1$ norm.

From the probability density function of a Poisson random variable, the likelihood writes: 
\begin{equation}
  \label{eq:7}
  p(y|x) = \prod_i \frac{((\Hmb x)[i])^{y[i]} \exp\left(-(\Hmb x)[i]\right)}{y[i]!}~,
\end{equation}
and the associated log-likelihood function is
\begin{equation}
  \label{eq:8}
  \ell\ell (y|x) = \sum_i \big( y[i]\log((\Hmb x)[i]) - (\Hmb x)[i] - \log(y[i]!)  \big)~.
\end{equation}
These formula are extended to the case $y=0$, using the convention $0!=1$. Taking negative
log-likelihood, we arrive at the following data fidelity term:
\begin{align}
  \label{eq:9}
  f_1\ &: \eta \in \mathbb{R}^n \mapsto \sum_{i=1}^n f_{\mathrm{poisson}}(\eta[i]), \\
  \text{if } y[i] > 0,\quad
  f_{\mathrm{poisson}}(\eta[i]) &=
  \begin{cases}
    -y[i] \log(\eta[i]) + \eta[i] & \text{if } \eta[i] > 0,\\
    +\infty & \text{otherwise,}
  \end{cases} \nonumber \\
  \text{if } y[i] = 0,\quad
  f_{\mathrm{poisson}}(\eta[i]) &=
  \begin{cases}
    \eta[i] & \text{if } \eta[i] \in [0,+\infty), \\
    +\infty & \text{otherwise.}
  \end{cases} \nonumber
\end{align}

Our aim is then to solve the following optimization problems, with either an analysis-type prior,
\begin{equation}
  \label{eq:11}
  \begin{gathered}
    \tag{$\Pm_{\eqi,\psi}$} \argmin_{\vx\in\mathbb{R}^n} J(\vx), \\
    J\ :\ \vx \mapsto {f_1\circ\Hmb(\vx)} + \underbrace{\eqi
      \Psi\circ\Phb\trans(\vx)}_{f_2\circ\Phb\trans(\vx)} + \underbrace{\imath_{\Cc} (\vx)}_{f_3(\vx)}~,
  \end{gathered}
\end{equation}
or a synthesis-type prior,
\begin{equation}
  \label{eq:12}
  \begin{gathered}
    \tag{$\Pm_{\eqi,\psi}'$} \argmin_{\va\in\mathbb{R}^L} J'(\va), \\
    J'\ :\ \va \mapsto {f_1\circ\Hmb\circ\Phb(\va)} + \underbrace{\eqi
      \Psi(\va)}_{f_2(\va)} + \underbrace{\imath_{\Cc} \circ
      \Phb(\va)}_{f_3\circ\Phb(\va)}~.
  \end{gathered}
\end{equation}
The penalty function $\Psi$ is additive, i.e. $\Psi(\va) = \sum_{i=0}^{L} \psi(\va[i])$, $\eqi > 0$ is a regularization parameter and $\imath_{\Cc}$ is the indicator function of the closed convex set $\Cc$. In our case, $\Cc$ is the positive orthant since we are fitting Poisson intensities, which are positive by nature.

In \eqref{eq:11}, the solution image $x_{\text{A}}^\star$ is directly targeted whose transform (analysis) coefficients $\Phb\trans x_{\text{A}}^\star$ are the sparsest. While in problem \eqref{eq:12}, we are seeking a sparse set of coefficients $\va_S^\star$ and the solution signal or image is synthesized from these representation coefficients and the dictionary $\Phb$ as $x_{\text{S}}^\star=\Phb\va_{\text{S}}^\star$. 

Obviously, problems \eqref{eq:11} and \eqref{eq:12} are not equivalent in general unless $\Phb$ is an orthobasis. For overcomplete non-orthogonal dictionaries, the solutions to synthesis and analysis formulations are different. Indeed, in the synthesis prior the set of  solutions $x_{\text{S}}^\star$ is confined to the column space of the dictionary, whereas in the analysis formulation, the solutions $x_{\text{A}}^\star$ are allowed to be arbitrary vectors in $\mathbb{R}^n$. Furthermore, for redundant dictionaries, there are much fewer unknowns in the analysis formulation, hence leading to a "simpler" optimization problem, although the proximity operators become more complicated because of composition even for simple functions. As opposed to the analysis approach, the synthesis approach has a constructive form providing an explicit description of the signals it represents, and as such, can benefit from higher redundancy to synthesize rich and complex signals. On the other hand, in the analysis approach, we can ask a signal or image to simultaneously agree with many columns of $\Phb$. This might become impossible with a highly redundant dictionary.

It can be shown that for a Parseval tight frame (PTF) dictionary $\Phb$, using the change of variable $\va=\Phb\trans x$, an analysis-type prior formulation can be written as a linearly constrained synthesis-type prior formulation. The constraint ensures that $\va$ remains in the span of $\Phb\trans$. Therefore, analysis-type prior forms can be seen as a subset of synthesis-type prior ones; see \cite{EladAnalysisSynthesis07} for another geometrical argument. If a global minimizer of the synthesis-type prior satisfies the span analysis constraint, i.e. is a feasible analysis solution, then both problems are equivalent. Other equivalence conditions can be drawn if stronger assumptions (but useless to promote sparse solutions) are imposed on the penalty $\Psi$; see e.g. \cite{EladAnalysisSynthesis07}.

With the notable exception of \cite{EladAnalysisSynthesis07}, and a recent paper by \cite{CandesAnalysis} on compressed sensing (with PTF dictionaries), a little is known about the actual theoretical guarantees of analysis-type priors in general. From a practical point of view, there is no consensus either on the type of conclusions to be drawn from the experimental work reported in the literature. Some authors have reported results where the analysis prior is superior to its synthesis counterpart, e.g. \cite{EladAnalysisSynthesis07,CandesIRL108,SelesnickSPIE09}. Others have shown that they are more or less equivalent. In our work here, the synthesis prior turns out to be better on the Saturn image, presumably because the dictionary is very well suited to sparsely synthesize the image (see Section~\ref{sec:imaging-dataset}).

In short, these phenomena are still not very well understood, especially for general inverse problems operators and dictionaries $\Phb$. More investigation is needed in this direction to extricate the deep distinctions and similarities between the two priors.

\subsection{Properties of the objective functions}
\label{sec:prop-object-funct}

From \eqref{eq:11} and \eqref{eq:12}, we have the following,
\begin{propo}
\label{prop:objectives}
  {~}\\ \vspace{-0.5cm}
  \begin{enumerate}[(i)]
  \item $f_1$ is a convex function and so are $f_1 \circ \Hmb$ and $f_1\circ\Hmb\circ\Phb$. 
  \item $f_1$ is strictly convex if $\forall i \in
    \{1,\ldots,n\}, y[i] \ne 0$. $f_1\circ\Hmb$ remains strictly convex if $\kerm(\Hmb) = \emptyset$, and so does $f_1\circ\Hmb\circ\Phb$ if $\Phb$ is an orthobasis and $\kerm(\Hmb) = \emptyset$.
  \item The gradient of $f_1$ is
   \begin{align}
      \label{eq:104}
      \nabla f_1(\eta) &= (g(\eta[i]))_{1\le i\le n},\\
      \text{if } y[i] > 0,\quad
      g(\eta[i]) &=
      \begin{cases}
        1 - \frac{y[i]}{\eta[i]} & \text{if } \eta[i] > 0, \\
        +\infty & \text{else}.
      \end{cases} \nonumber \\
      \text{if } y[i] = 0,\quad
      g(\eta[i]) &=
      \begin{cases}
        1 & \text{if } \eta[i] \ge 0, \\
        +\infty & \text{else}.
      \end{cases} \nonumber  
  \end{align}
    \item Suppose that $(0,+\infty) \cap \Hmb\left([0,+\infty)\right) \neq \emptyset$ and $\Psi$ has full domain. Then both $J \in \Gamma_0(\mathbb{R}^n)$ and $J' \in \Gamma_0(\mathbb{R}^L)$.
  \end{enumerate}
\end{propo}
\begin{proof}
  The proof of (i) and (ii) follow from well-known properties of convex functions. (iii) is obtained by straightforward
  calculations. (iv): $J$ and $J'$ are the sums of lsc convex functions which entails their convexity and lower
  semicontinuity. $\Phb$ is a frame, hence surjective. Thus under the assumptions of (iv), the domains of both $J$ and
  $J'$ are non-empty, i.e. proper functions.
\end{proof}

\subsection{Well-posedeness of \eqref{eq:11} and \eqref{eq:12}}
\label{sec:char-solut}

Let $\Mc$ and $\Mc'$ be respectively the sets of minimizers of problems \eqref{eq:11} and \eqref{eq:12}. Suppose that $\exists i \in \{1,\ldots,n\}$ such that $y[i] \ne 0$, otherwise the minimizer would be trivially 0. Thus $J$ and $J'$ are coercive ($\Phb\Tr$ is injective). Therefore, the following holds:
\begin{propo}
  {~} \\ \vspace{-0.5cm}
  \begin{enumerate}
  \item Existence: each of \eqref{eq:11} and \eqref{eq:12} has at least one solution, i.e. $\Mc\ne\emptyset$ and $\Mc'\ne\emptyset$.
  \item Uniqueness: \eqref{eq:11} and \eqref{eq:12} have unique solutions if $\Psi$ is strictly convex, or under (ii) of Proposition~\ref{prop:objectives}.
  \end{enumerate}
\end{propo}



\section{Iterative minimization algorithm}
\label{sec:minimizationalgorithm}

We are now ready to describe the proximal splitting algorithm to solve \eqref{eq:11} and \eqref{eq:12}. At the heart of
this splitting framework is the notion of proximity operator which we are about to introduce.

\subsection{Proximity operator}
The proximity operator is a generalization of the projection onto a closed convex set, introduced by
Moreau~\cite{Moreau1962,Moreau1963,Moreau1965},
\begin{definition}[Moreau\cite{Moreau1962}]
  \label{def:1} 
  Let $f \in \Gamma_{0}(\Hc)$. Then, for every $x\in\Hc$, the function $y \mapsto
  f(y) + \norm{x-y}^{2}/2$ achieves its infimum at a unique point denoted by
  $\prox_{f}x$. The operator $\prox_{f} : \Hc \to \Hc$ thus defined is the
  \textit{proximity operator} of $f$.  
\end{definition}

The proximity operator enjoys many properties, among these, one can easily show that $\forall x,p \in \Hc$
  \begin{align}
    \label{eq:prox}
    p = \prox_{f} x \iff x-p \in \partial f(p) \iff \pds{y-p}{x-p} + f(p) \le f(y) ~ \forall y\in\Hc ~,
  \end{align}
  which means that $\prox_{f}$ is the resolvent of the subdifferential of
  $f$ \cite{Eckstein92}. Recall that the resolvent of the subdifferential
  $\partial f$ is the single-valued operator $J_{\partial f}: \Hc \to \Hc$ such
  that $\forall x \in \Hc, x - J_{\partial f}(x)
  \in \partial f(J_{\partial f}) \iff J_{\partial f} = (\Id
  + \partial f)^{-1}$.

\subsection{A splitting algorithm for sums of convex functions}
\label{sec:gener-minim-algor}


Suppose that the objective function can be expressed as the sum of $K$ convex, lsc and proper functions,
\begin{equation}
  \label{eq:5}
  \argmin_{x \in \Hc} ~ (f(x) = \sum_{i=1}^K f_i(x))~.
\end{equation}
Proximal splitting methods for solving \eqref{eq:5} are iterative algorithms which may evaluate the individual proximity
operators $\prox_{f_i}$ (perhaps approximately), but never proximity operators of sums of the $f_i$, and a fortiori that
of $f$.  The key idea relies on the formulation of \eqref{eq:5} which is such that the proximity operator of each
individual function has some (relatively) convenient structure, while those of the sums do not in general. This turns
out to be true in our case for problems \eqref{eq:11} and \eqref{eq:12}.

Splitting algorithms have an extensive literature since the 1970's, where the case $K=2$ predominates. Usually,
splitting algorithms handling $K > 2$ have either explicitly or implicitly relied on reduction of \eqref{eq:5} to the
case $K = 2$ in the product space $\Hc^K$. For instance, applying the Douglas-Rachford splitting to the reduced form
produces Spingarn's method \cite{Spingarn}, which performs independent proximal steps on each $f_i$, and then computes
the next iterate by essentially averaging the individual proximity operators. The method proposed in
\cite{Combettes2008} is very similar in spirit to Spingarn's method, where moreover, it allows for inexact computation
of the proximity operators.

For every $i \in \{1,\ldots,K\}$, let $(a_{t,i})_{t\in\mathbb{N}}$ be a sequence in $\Hc$. Let $(x_t)_{t\in\mathbb{N}}$
be the sequence as constructed by Algorithm~\ref{algo:proxgen}.  The following result due to \cite{Combettes2008}
establishes the convergence of $(x_t)_{t\in\mathbb{N}}$. We reproduce it here for the sake of completeness.
\begin{theorem}[\cite{Combettes2008}]
  \label{th:convalg}
  Suppose that $f$ is coercive and let $(x_t)_{t\in\mathbb{N}}$ be a sequence
  generated by Algorithm~\ref{algo:proxgen} under the following assumptions,
  \begin{enumerate}
  \item $f$ is a proper function;
  \item $(0,\ldots,0)\in \mathrm{ri}\left\{ (x-x_1,\ldots,x-x_K)\mid x\in\Hc,x_1\in\dom f_1,\ldots,x_K\in\dom f_K\right\}$;
  \item $\sum_{t\in\mathbb{N}} \theta_t(2-\theta_t) = +\infty$;
  \item $\forall i \in  \{1,\ldots,K\}\quad \sum_{t\in\mathbb{N}} \theta_t \norm{a_{(t,i)}} < +\infty$.
  \end{enumerate}
  Then $(x_t)_{t\in\mathbb{N}}$ converges to a (non-strict) global minimizer.
\end{theorem}

\begin{algorithm}[htbp]
  \noindent{\bf{Task:}} Solve the convex optimization problem~\eqref{eq:5}.\\
  \noindent{\bf{Initialization:}} \\
  Choose $\mu \in (0,+\infty)$, $(p_{(0,i)})_{1\le i \le K} \in \Hc^K$,
  $(\omega_i)_{1\le i\le K} \in (0,1]^K \text{ satisfying } \sum_{i=1}^K \omega_i = 1$ and
  $x_0 = \sum_{i=1}^K \omega_i p_{(0,i)}$. \\
  \noindent{\bf{Main iteration:}} \\
  \noindent{\bf{For}} $t=0$ {\bf{to}} $N_{\mathrm{ext}}-1$,
  \begin{align*}
    & \forall i \in  \{1,\ldots,K\}\quad \xi_{(t,i)} = \prox_{\mu f_i/\omega_i} p_{(t,i)} + a_{(t,i)}~, \\
    & \xi_{t} = \sum_{i=1}^K \omega_i \xi_{(t,i)}~, \\
    & \text{Choose } \theta_t \in ]0,2[~, \\
    & \forall i \in  \{1,\ldots,K\}\quad p_{(t+1,i)} = p_{(t,i)} + \theta_t\left( 2\xi_t - x_t - \xi_{(t,i)} \right)~, \\
    & x_{t+1} = x_t + \theta_t (\xi_t - x_t)~.
  \end{align*}
  \noindent{\bf{End main iteration}} \\
  \noindent{\bf{Output:}} A solution of \eqref{eq:5}: $x_{N_{\mathrm{ext}}}$.
  \caption{}
  \label{algo:proxgen}
\end{algorithm}

Notice that the sequences $(a_{(t,i)})_{t\in\mathbb{N}}$ allow for some errors in the evaluation of the different
proximity operators, with the proviso that these errors remain summable. This is useful when the closed forms of the proximity operators are difficult to construct and rather, must be computed via an inner iteration. This will turn to be the case for some terms involved in the two objective functions of \eqref{eq:11} and \eqref{eq:12}. More precisely, the difficulty at stake is how to deal with the composition with a bounded linear operator, here the circular convolution operator and/or the dictionary. This is what we are about to handle.

\subsection{Proximity operator of a pre-composition with a linear operator}
\label{sec:prox-oper-funct}

The following result will manifest its importance as a building block to solve the two problems \eqref{eq:11} and \eqref{eq:12}. Indeed, it expresses the proximity operator of a function $f\in\Gamma_0(\Hc)$ composed with an affine operator $\Amb : \RR^n \to \RR^m, x \mapsto \Fmb x - y$, $y\in\RR^m$ where $\Fmb: \RR^n \to \RR^m$ is a bounded linear operator.
\begin{theorem}
  {~}\\
  \label{th:compo}
  Let $\Fmb$ be a linear bounded operator such that the domain qualification condition $\ri(\dom(f) \cap \mathrm{Im}(\Amb)) \ne
  \emptyset$ holds. Then  $f \circ \Amb \in \Gamma_0(\RR^m)$ and
  \begin{enumerate}[(i)]
  \item If $\Fmb$ is a tight frame. Then,
    \begin{equation}
      \label{eq:86}
      \prox_{f\circ \Amb}(x) = y + c^{-1}\Fmb\trans(\prox_{cf}-\Id) \Amb(x).
    \end{equation}
  \item If $\Fmb$ is a general frame. Let $\desc_t \in (0,2/c_2)$. Let $(u_t)_{t\in\mathbb{N}}$ be sequence of iterates provided by Algorithm~\ref{algo:proxalg1}. Then, $(u_t)_{t\in\mathbb{N}}$ converges to $\bar{u}$
    and $(p_t)_{t\in\mathbb{N}}$ converges to $\prox_{f\circ\Amb} x = x -
    \Fmb\trans\bar{u}$. More precisely, these two sequences converge linearly and the best
    rate is attained for $\desc_t \equiv 2/(c_1+c_2)$:
    \begin{equation}
      \label{eq:52}
      \norm{p_t - \prox_{f\circ\Amb}(x)} \le \sqrt{\frac{c_2}{c_1}} \left(\frac{c_2-c_1}{c_2+c_1}\right)^t
      \norm{p_0 - \prox_{f\circ\Amb}(x)}~.
    \end{equation}
  \item In all other cases, apply Algorithm~\ref{algo:proxalg1}
    with $\desc_t\in(0,2/c_2)$. Then, $(u_t)_{t\in\mathbb{N}}$ converges to $\bar{u}$, and
    $(p_t)_{t\in\mathbb{N}}$ converges to $\prox_{f\circ\Amb} x = x - \Fmb\trans\bar{u}$ at the rate
    $\Oc(1/t)$.
    \end{enumerate}
\end{theorem}
See the appendix for a concise proof.\\

\begin{algorithm}[t]
  \noindent{\bf{Task:}} Forward-backward algorithm to compute $\prox_{f \circ \Amb}(x)$. \\
  \noindent{\bf{Parameters:}} The function $f$, the linear bounded operator $\Amb$, number of iterations $N_{\mathrm{int}}$ and step-size $\desc_t\in(0,2/c_2)$. \\
  \noindent{\bf{Initialization:}}
  Choose $u_0 \in\mathbb{R}^m$, $p_0=x-\Fmb\trans u_0$. \\
  \noindent{\bf{Main iteration:}} \\
  \noindent{\bf{For}} $t=0$ {\bf{to}} $N_{\mathrm{int}}-1$,
  \begin{equation}
    \label{eq:1}
    \begin{split}
      u_{t+1} &= \desc_t (\Id - \prox_{f/\desc_t})(u_t/\desc_t + \Amb p_t), \\
      p_{t+1} &= x - \Fmb\trans u_{t+1}.
    \end{split}
  \end{equation}
  \noindent{\bf{End main iteration}} \\
  \noindent{\bf{Output:}} The proximity operator of $f\circ\Amb$ at $x$~: $p_{N_{\mathrm{int}}}$~.
  \caption{}
  \label{algo:proxalg1}
\end{algorithm}

When $\Fmb$ is a frame, the convergence speed of the one-step forward-backward scheme \cite{Gabay83,FadiliStarckBook09} depends clearly on the redundancy of the frame. The higher the redundancy, the slower the convergence, i.e. the number of necessary iterations to obtain an $\varepsilon$-accurate solution is $\mathcal{O}\left(\frac{c_1}{c_2}\log \varepsilon^{-1}\right)$. For the general case (iii) where $\Fmb$ is not a frame, the algorithm necessitates as large as $\mathcal{O}(1/\varepsilon)$ iterations to reach an $\varepsilon$-accuracy on the iterates.

Assessing the convergence rate as we have just done is important when this algorithm will be used in a nested scheme as a sub-iteration to compute the proximity operator of the composition with a linear operator. Indeed, Theorem~\ref{th:convalg} and the discussion thereafter, clearly states the stability of Algorithm~\ref{algo:proxgen} to errors in the individual proximity operators of the functions $F_i, ~ i=1,\cdots,K$, provided that the errors remain summable. In a nutshell, this means that if Algorithm~\ref{algo:proxalg1} is used within Algorithm~\ref{algo:proxgen}, a sufficient number of inner iterations $N_{\mathrm{int}}$ must be chosen to ensure error summability. This number of inner iterations obviously depends on the convergence speed of Algorithm~\ref{algo:proxalg1}, hence on the structure of $\Fmb$.

\subsection{Proximity operator of a sparsity-promoting penalty}
\label{sec:prox-oper-spars}

To implement the above iteration, we need to express $\prox_{\eqi\psi}$.  The following
result gives us the formula for a wide class of penalties $\psi$ :
\begin{lemma}
  \label{th:3}
  Suppose that $\psi$ satisfies, (i) $\psi$ is convex even-symmetric , non-negative and
  non-decreasing on $[0,+\infty)$, and $\psi(0)=0$. (ii) $\psi$ is twice differentiable on
  $\mathbb{R}\setminus \{0\}$. (iii) $\psi$ is continuous on $\mathbb{R}$, it is not
  necessarily smooth at zero and admits a positive right derivative at zero $\psi^{'}_+(0) =
  \lim_{h\to 0^+} \frac{\psi(h)}{h} > 0$. Then, the proximity operator of
  $\eqi\Psi(\va)$, $\prox_{\eqi\Psi}(\va)$ has exactly one continuous solution
  decoupled in each coordinate $\va_i$ :
  \begin{equation}
    \label{eq:10}
    \prox_{\eqi\psi}(\va_i) =
    \begin{cases}
      0 & \text{if } \abs{\va_i} \le \eqi\psi^{'}_+(0)\\
      \va_i-\eqi\psi^{'}(\bar{\va}_i) & \text{if } \abs{\va_i} > \eqi\psi^{'}_+(0)
    \end{cases}
  \end{equation}
\end{lemma}
A proof of this lemma can be found in \cite{Fadili2006}. A similar result is also proposed
in \cite{Combettes2007}. Among the most popular penalty functions $\psi$ satisfying the
above requirements, we have $\psi(\va_i) = \abs{\va_i}$, in which case the associated
proximity operator is the popular soft-thresholding.

\subsection{Proximity operator of the data fidelity $f_1$}
\label{sec:prox-fidelity}
The following result can be proved easily by solving the proximal optimization problem in Definition~\ref{def:1} with $f_1$ as defined in \eqref{eq:9}, see also \cite{Combettes2007a}.

\begin{lemma}
   \label{lem:prpois}
   Let $\vy$ be the count map (i.e. the observations), the proximity operator associated to $f_1$ (i.e. the Poisson anti
   log-likelihood) is,
   \begin{equation}
     \label{eq:3}
      \prox_{\beta f_1} \vx = \left( 
        \frac{\vx[i] - \beta + \sqrt{(\vx[i] -\beta)^2 + 4\beta \vy[i]}}{2}
        \right)_{1\le i \le n}~.
    \end{equation}
  \end{lemma}

\section{Sparse iterative deconvolution}
\label{sec:sparse-iter-deconv}

The game to play now is to plug the previous results to compute the proximity operators involved in problems \eqref{eq:11} and \eqref{eq:12} in Algorithm~\ref{algo:proxgen}.

\begin{algorithm}[h]
  \noindent{\bf{Task:}} Image deconvolution with Poisson noise, solve \eqref{eq:12}. \\
  \noindent{\bf{Parameters:}} The observed image counts $y$, the dictionary $\Phb$, number of main iterations
  $N_{\mathrm{ext}}$, number of sub-iterations $N_{\mathrm{int}}$ for the proximal operator of the data fidelity term, $\mu > 0$ and regularization parameter $\eqi$. \\
  \noindent{\bf{Initialization:}}\\
  $\forall i \in \{0,1,2\},\quad p_{(0,i)} = \Phb\trans\vy$. \\
  $\va_0 = \Phb\trans y$. \\
  \noindent{\bf{Main iteration:}} \\
  \noindent{\bf{For}} $t=0$ {\bf{to}} $N_{\mathrm{ext}}-1$,
  \begin{itemize}
  \item \underline{Data fidelity} (Lemma~\ref{lem:prpois} and Theorem~\ref{th:compo}(i)-(iii)): $\xi_{(t,0)} =
    \prox_{\mu f_1\circ\Hmb\circ\Phb/3} p_{(t,0)}$.
  \item \underline{Sparsity-penalty} (Lemma~\ref{th:3}): $\xi_{(t,1)} =
    \prox_{\mu \eqi\Psi/3} p_{(t,1)} = \mathrm{ST}_{\mu\eqi/3}(p_{(t,1)})$.
  \item \underline{Positivity constraint} (Theorem~\ref{th:compo}(i)): $\xi_{(t,2)} = \prox_{\imath_\Cc\circ\Phb/3}(p_{(t,3)})$.
  \item Average the proximity operators: $\xi_{t} = (\xi_{(t,0)} + \xi_{(t,1)} + \xi_{(t,2)})/3$. 
  \item Choose $\theta_t\in]0,2[$.
  \item Update the components: $\forall i \in \{0,1,2\},\quad p_{(t+1,i)} = p_{(t,i)} + \theta_t (2\xi_{t} - \va_{t} - \xi_{(t,i)})$.
  \item Update the coefficients estimate: $\va_{t+1} = \va_t + \theta_t(\xi_t - \va_t)$
  \end{itemize}
  \noindent{\bf{End main iteration}} \\
  \noindent{\bf{Output:}} Deconvolved image $x^{\star}=\Phb\va_{N_{\mathrm{ext}}}$.
  \caption{}
  \label{algo:deconv2}
\end{algorithm}

\subsection{Analysis-type prior}
\label{sec:analysis-prior}

For problem \eqref{eq:11}, the different proximity operator are computed as follows:
\begin{enumerate}[(1)]
\item $\prox_{f_1\circ\Hmb}$ is obtained owing to Theorem \ref{th:compo}(iii) and Lemma~\ref{lem:prpois};
\item $\prox_{f_2\circ\Phb\trans} = \prox_{\eqi\Psi\circ\Phb\trans}$ is given by Theorem \ref{th:compo}(iii) and Lemma \ref{th:3};
\item $\prox_{f_3} = \prox_{\imath_{\Cc}}$ is directly is the euclidean projection onto
  the closed convex set $\Cc$. 
\end{enumerate}

\subsection{Synthesis-type prior}
\label{sec:synthesis-prior}

As far as problem \eqref{eq:12} is concerned, here is how the proximity operators are computed:
\begin{enumerate}[(1)]
\item $\prox_{f_1\circ\Hmb\circ\Phb}$ is computed using Theorem~\ref{th:compo} and
  Lemma~\ref{lem:prpois}: use Theorem~\ref{th:compo}(i) and Theorem~\ref{th:compo}(iii) if $\Phb$ is a tight frame, or Theorem~\ref{th:compo}(iii) otherwise;
\item $\prox_{f_2} = \prox_{\eqi\Psi}$ is given using Lemma \ref{th:3};
\item $\prox_{f_3\circ\Phb} = \prox_{\imath_{\Cc}\circ\Phb}$ is given by Theorem~\ref{th:compo}(i) or (ii) and the projector onto the closed convex set $\Cc$: use Theorem~\ref{th:compo}(i) if $\Phb$ is a tight frame, or Theorem~\ref{th:compo}(ii) if it is a general frame;
\end{enumerate}
For example, let $f_2$ be the $\ell_1$-norm (i.e. $\psi = \abs{.}$) and $f_3 = \imath_{\Cc} \circ \Phb$, where $\Cc$ is the positive orthant and $\Phb$ is a tight frame, then \eqref{eq:12} is solved by Algorithm~\ref{algo:deconv2}.

\subsection{Computational complexity and implementation details}
\label{sec:comp-compl-impl}

The bulk of computation of our deconvolution algorithm is invested in applying $\Phb$ (resp. $\Hmb$) and its adjoint
$\Phb\trans$ (resp. $\Hmb\trans$). These operators are never constructed explicitly, rather they are implemented as fast
implicit operators taking a vector $x$, and returning $\Phb x$ (resp. $\Phb\trans x$) and $\Hmb x$ (resp. $\Hmb\trans
x$). Multiplication by $\Hmb$ or $\Hmb\trans$ costs two FFTs, that is $2n \log n$ operations ($n$ denotes the number of
pixels). The complexity of $\Phb$ and $\Phb\trans$ depends on the transforms in the dictionary: for example, the
orthogonal wavelet transform costs $\mathcal{O}(n)$ operations, the translation-invariant discrete wavelet transform
(TI-DWT) costs $\mathcal{O}(n\log n)$, the curvelet transform costs $\mathcal{O}(n\log n)$, etc. Let $V_{\Phb}$ denote
the complexity of applying the analysis or synthesis operator. Define $N_{\mathrm{ext}}$ and $N_{\mathrm{int}}$ as the
maximal number of outer and inner iterations respectively in Algorithm~\ref{algo:proxgen} and Algorithm~\ref{algo:proxalg1}, and recall that $L$ is the number of coefficients. The computational complexity of the algorithms for solving \eqref{eq:11} and \eqref{eq:12}
are summarized  as follows:
\begin{center}
  \begin{tabular}{|c|c|} 
    \hline Problem & Computational complexity \\
    \hline \hline & $\Phb$ orthobasis or tight frame \\
    \hline \eqref{eq:11} & $N_{\mathrm{ext}}\left(N_{\mathrm{int}} \left( 4\Oc(n\log n) + 2 V_{\Phb}\right) + \Oc(n) \right)$ \\
    \hline \eqref{eq:12} & $N_{\mathrm{ext}}\left(4V_{\Phb} + N_{\mathrm{int}} \left(4n\log n\right) + \Oc(L) \right)$ \\

    \hline \hline & $\Phb$ general frame or linear bounded operator \\
   \hline \eqref{eq:11} & $N_{\mathrm{ext}}\left(N_{\mathrm{int}} \left( 4\Oc(n\log n) + 2 V_{\Phb}\right) + \Oc(n) \right)$ \\
    \hline \eqref{eq:12} & $N_{\mathrm{ext}}\left( N_{\mathrm{int}} \left(4V_{\Phb} + 4n\log n + \Oc(L)\right) \right)$ \\
    \hline
  \end{tabular}
\end{center}

\section{Results}
\label{sec:results}

\subsection{Experimental protocol}
\label{sec:prot-exper}
We applied the proposed algorithms on a simulated image of the Hubble Space Telescope Wide Field Camera of a distant
cluster of galaxies~\cite{Starck2006} and an image of the planet Saturn. In the sequel, our algorithms are respectively
dubbed Prox-FA (for the one with the analysis prior) and Prox-FS (for that associated to the synthesis prior). For both
of them, the number of interior iterations $N_{\mathrm{int}}$ was set to $10$ which was enough to ensure convergence
(i.e. the error terms in the computation of the proximity operators were under control). The obtained results were
compared with others state-of-the art algorithms: the naive approach treating the noise as if it were additive white
Gaussian (NaiveG~\cite{Vonesch2007}), the Anscombe variance stabilizing approach (StabG~\cite{Dupe2009c}), and the
algorithms proposed in \cite{Figueiredo2010} with the synthesis (PIDAL-FS) and the analysis (PIDAL-FA) priors. For all
compared algorithms, the sparsity-promoting penalty was the $\ell_1$-norm whose proximity operator is well-known
soft-thresholding operator. For fair comparison, all algorithms were stopped when the relative change between two
consecutive iterates was below a given tolerance $\delta > 0$, i.e.
\begin{equation}
  \label{eq:4}
  \norm{x_{t+1} - x_t}_2 / \norm{x_t}_2 \leq \delta~.
\end{equation}
The quantitative score of restoration quality was the the mean absolute error (MAE). The MAE was chosen as it is well suited to Poisson noise \cite{Dupe2009c}.

As usual in regularized inverse problems, the choice of the regularization parameter $\eqi$ is crucial. When the reference image is available, its choice can be tweaked manually so as to maximize the restoration quality by minimizing some quality score; e.g. MAE. But of course, this is a tedious task and cumbersome for real data. To circumvent these difficulties, in \cite{Dupe2009c}, an objective model selection driven choice based on the GCV was proposed. An approximate closed-form of the GCV was derived (see for \cite{Dupe2009c} details),  
\begin{align}
    \label{eq:6}
    \mathrm{GCV}(\eqi) = \frac{\norm{2\sqrt{y + \frac{3}{8}} - 2\sqrt{\Hmb \vx^{\star} + \frac{3}{8}}}^2}{(n - df)^2}~, \\
    \text{with } df \approx \card\left\{i=1,\ldots,L ~ \big|~ |\alpha^{\star}_i| \geq \eqi \right\}~, \nonumber
\end{align}
where $\vx^{\star}$ is the restored image provided by the estimator. It was shown in \cite{Dupe2009c} that the GCV-based choice of $\eqi$, although turned out to be close to the value that actually minimizes the MAE and the MSE, was slightly upper-biased in practice. The GCV then was used as a first good guess that might need refinement.



\subsection{Simulated sky image}
\label{sec:simulated-image}

Figure~\ref{fig:sky}(a) displays a simulated image of the sky. The maximal intensity (mean photon count) in this image
is $\sim$18000. The blurred version by $h$ (using the Hubble's PSF before being repaired \cite{Starck2006}) and the
noisy ones are depicted in Figure~\ref{fig:sky}(b) and (c). For this image, the dictionary $\Phb$ contained a tight
frame of translation-invariant wavelets (TI-DWT). For this image, we also include in the comparative study the
deconvolution obtained with Richardson-Lucy regularized with the wavelet multiresolution support (RL-MRS) as proposed in
\cite{Starck95}. The latter was specifically designed for this kind of objects. We also compare to the total-variation
regularized Richardson-Lucy algorithm (RL-TV) proposed in \cite{Dey2004}, and the fast translation invariant
tree-pruning reconstruction combined with an EM algorithm (FTITPR~\cite{Willett2004}). For each compared method, the
regularization parameter was tuned manually to achieve its best performance level (i.e. minimize MAE), and the convergence tolerance
$\delta$ was set to $10^{-5}$.

\begin{table}
  \centering
  {\footnotesize
    \begin{tabular}{|l||c|c|c|c|c|}\hline
      & RL-MRS \cite{Starck2006}& RL-TV \cite{Dey2004}& FTITPR \cite{Willett2004} & NaiveG \cite{Vonesch2007}& StabG \cite{Dupe2009c} \\
      \hline MAE & 63.5      & 52.8    & 60.9    & 39.8  & 43  \\

      \hline \hline & PIDAL-FA \cite{Figueiredo2010}& PIDAL-FS \cite{Figueiredo2010}& Prox-FA & \multicolumn{1}{c|}{Prox-FS} \\\cline{1-5} 
      MAE &  40           &  43.6        & 37.8        & \multicolumn{1}{c|}{35.6} \\\cline{1-5}
    \end{tabular}
  }
  \caption{MAE for the deconvolution of the sky image.}
  \label{tab:maesky}
\end{table}

For this image, the best algorithm is clearly RL-MRS where most of the small faint structures as well as bright objects
are well deconvolved. However, small artifacts from the wavelet transform are also present. The StabG algorithm does a
good job and preserves most of the weak objects which are barely visible in the degraded noisy image. At this intensity
regime, the NaiveG algorithm yielded satisfactory results, which are comparable to ours. FTITPR correctly restores most
of the important structures with a clean background, but many faint objects are lost. RL-TV leads to a deconvolution
similar to ours for bright objects, but the background is dominated by spurious artifacts. The PIDAL approaches give
quite similar results, PIDAL-FA leads here to a smoother estimate than PIDAL-FS. In the other hand PIDAL-FS seems to
preserve more sources, as it can be seen for example on the bottom left, where one can notice objects that do not appear
in the results of PIDAL-FA. In the same way, our algorithms yield results comparable to their respective parts of PIDAL
as expected since they solve the same optimization problems. Prox-FS shows one of the best deconvolution results with a
clean background and good restoration of most objects, although a few small details are oversmooth (see on bottom left).

These qualitative results are confirmed by a quantitative score of the restoration quality in terms of the MAE, see Table~\ref{tab:maesky}. Notice that the hight MAE value of RL-MRS, which might seem surprising given its good visual result, might be explained by the artifacts that locally destroy the photometry.

\begin{figure}[ht]
  \centering
  \begin{tabular}{@{ }c@{ }c@{ }c@{ }}
    \includegraphics[width=0.33\linewidth]{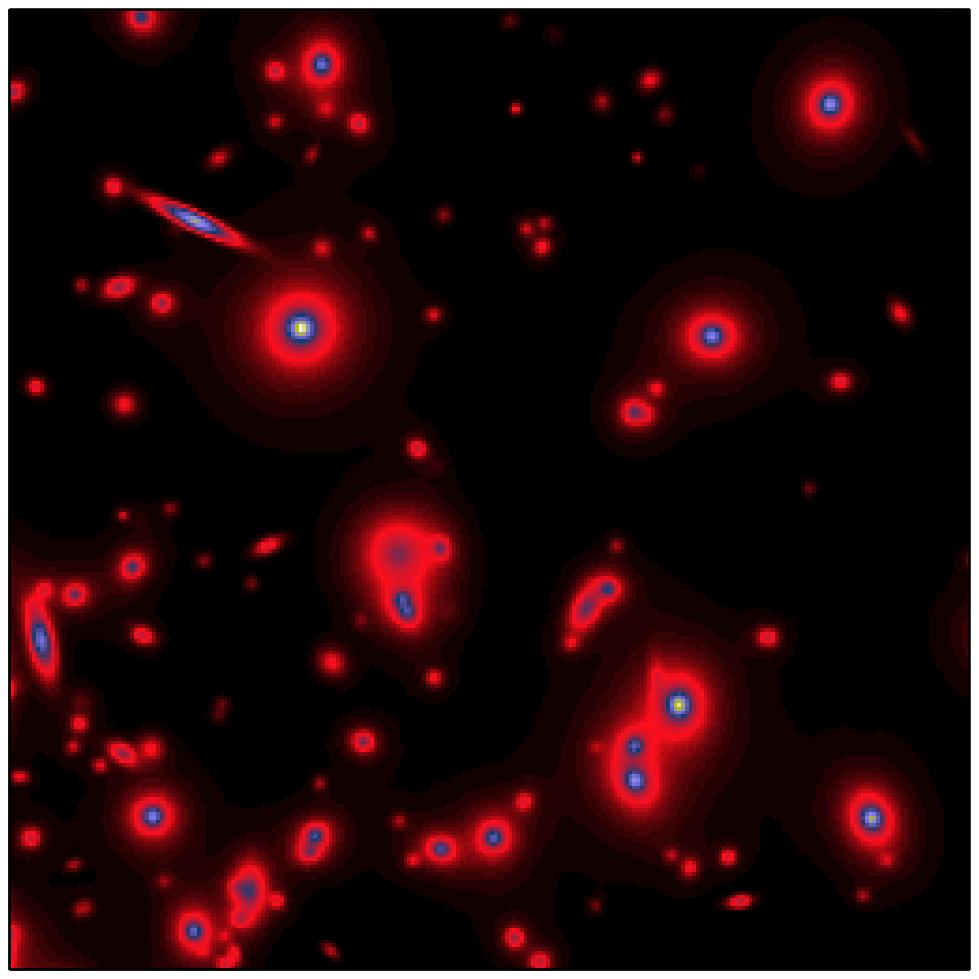} &
    \includegraphics[width=0.33\linewidth]{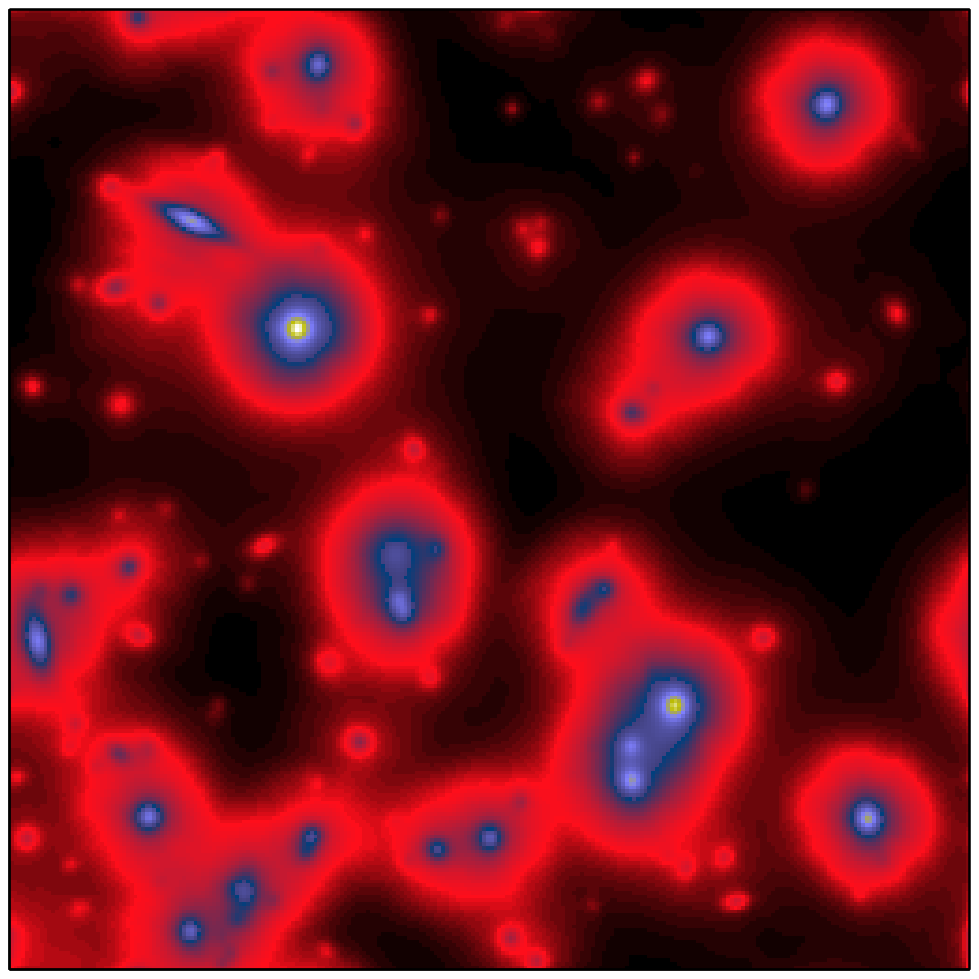} &
    \includegraphics[width= 0.33\linewidth]{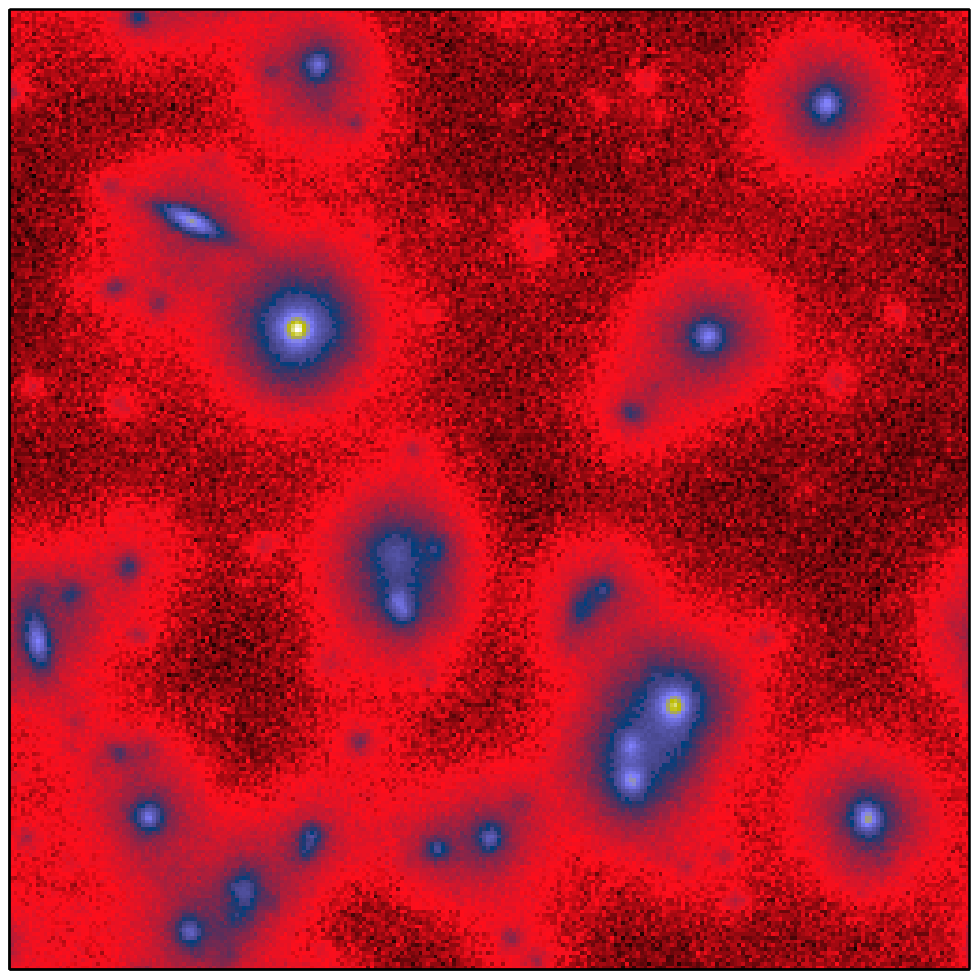} \\
    (a) & (b) & (c)	\\
    \includegraphics[width=0.33\linewidth]{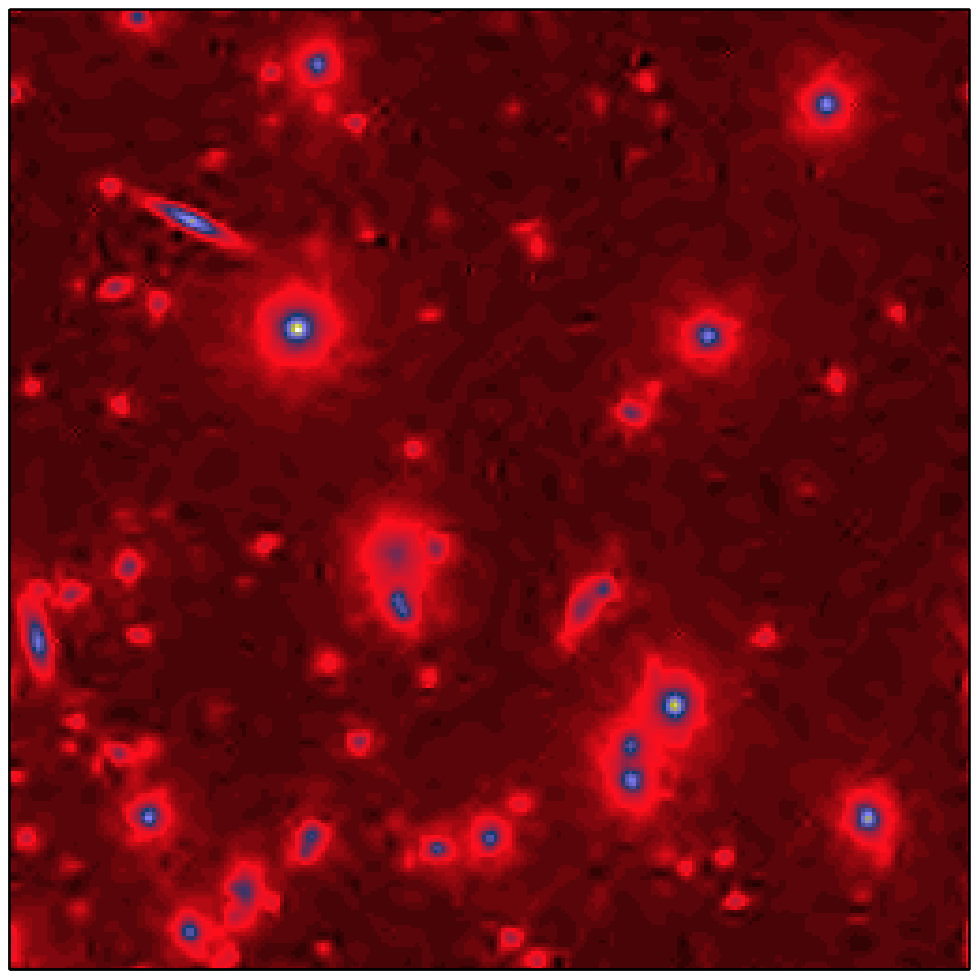} &
    \includegraphics[width=0.33\linewidth]{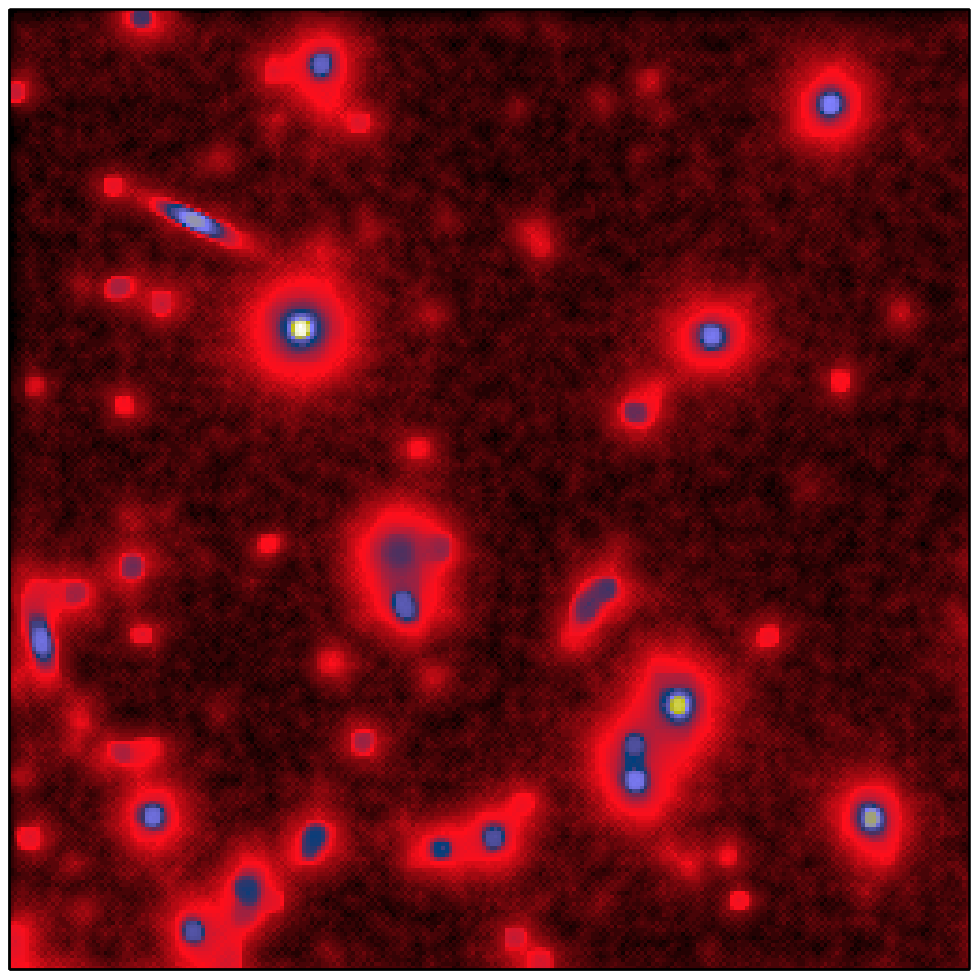} &
    \includegraphics[width=0.33\linewidth]{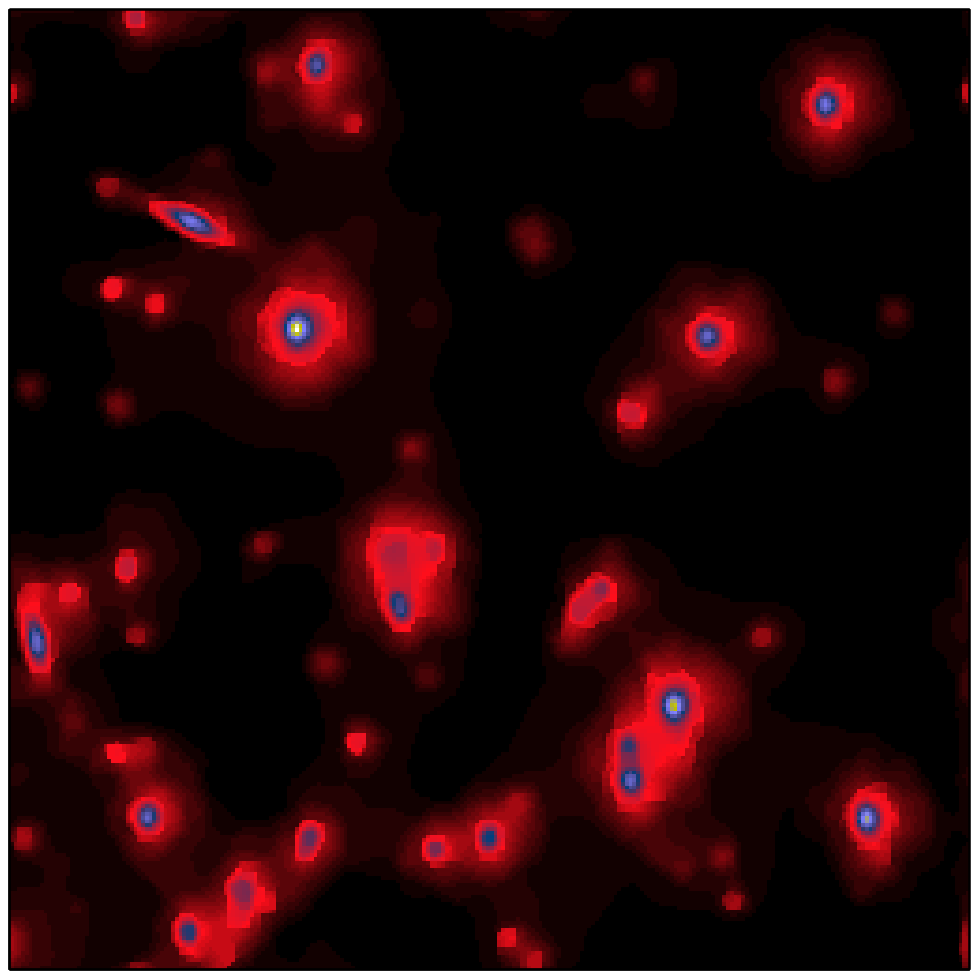} \\
    (d) & (e) & (f) \\
    \includegraphics[width=0.33\linewidth]{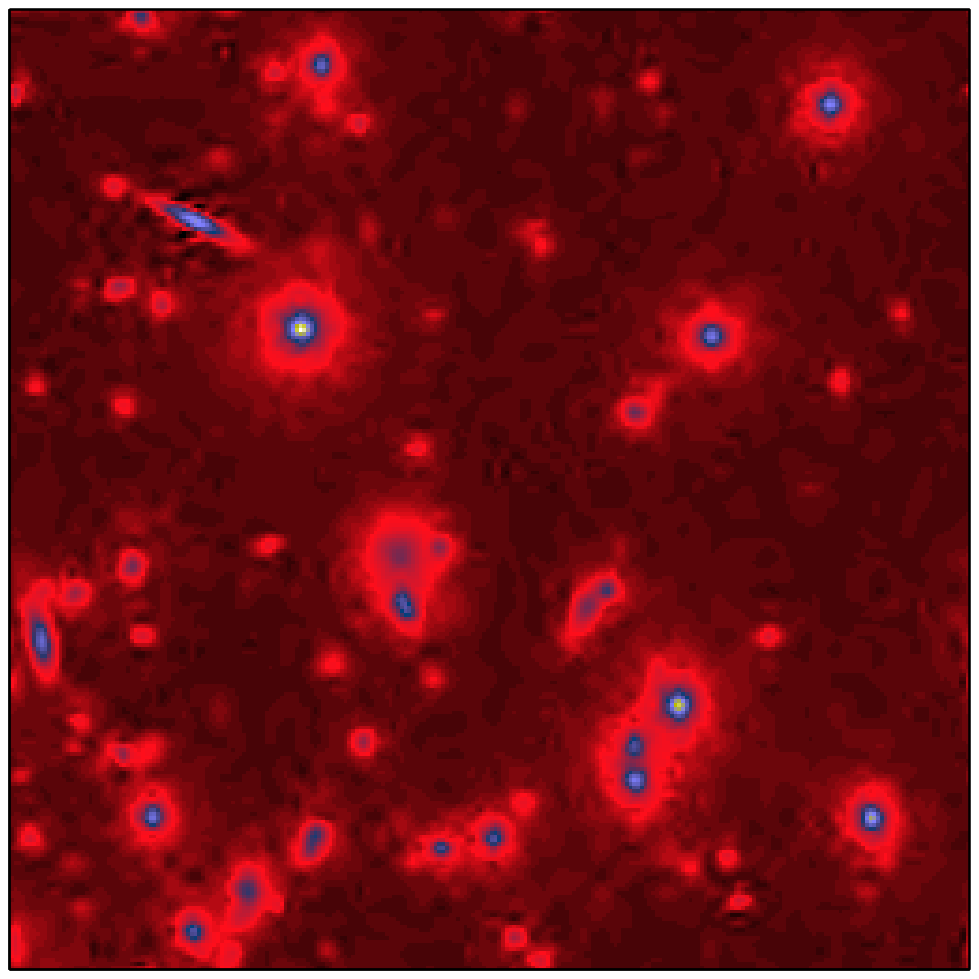} &
    \includegraphics[width=0.33\linewidth]{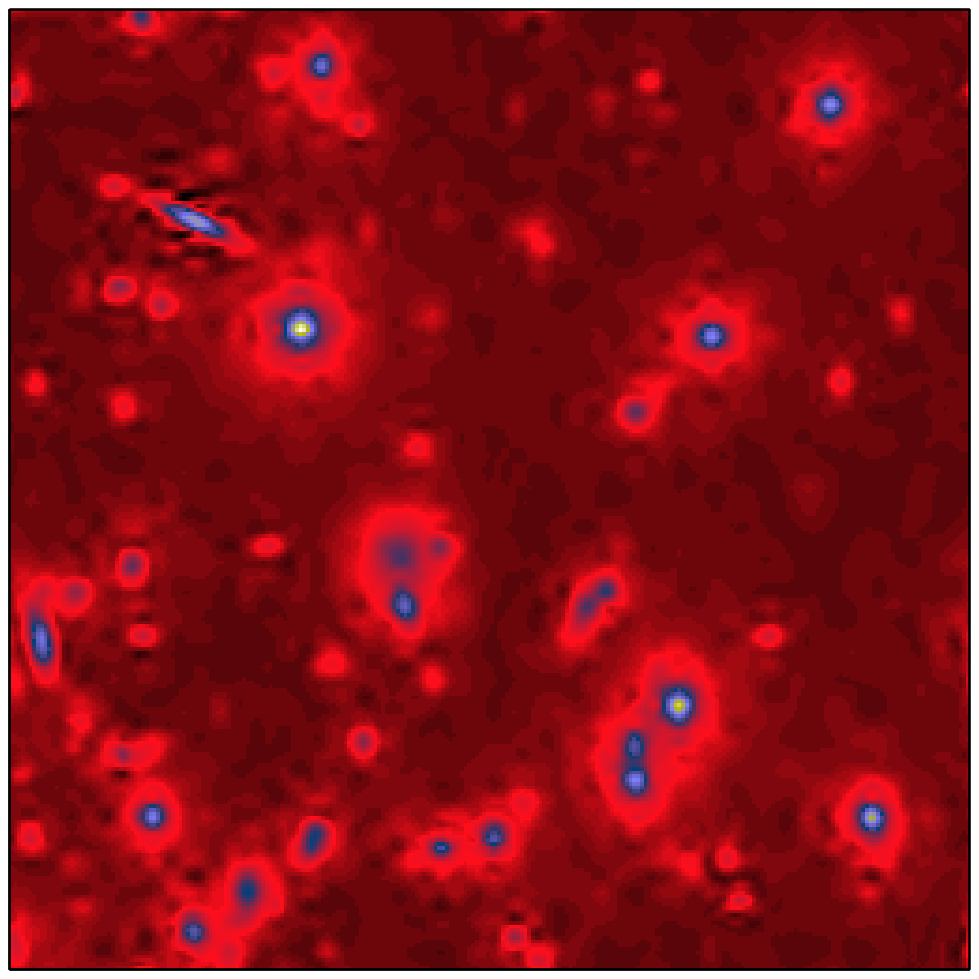} &
    \includegraphics[width=0.33\linewidth]{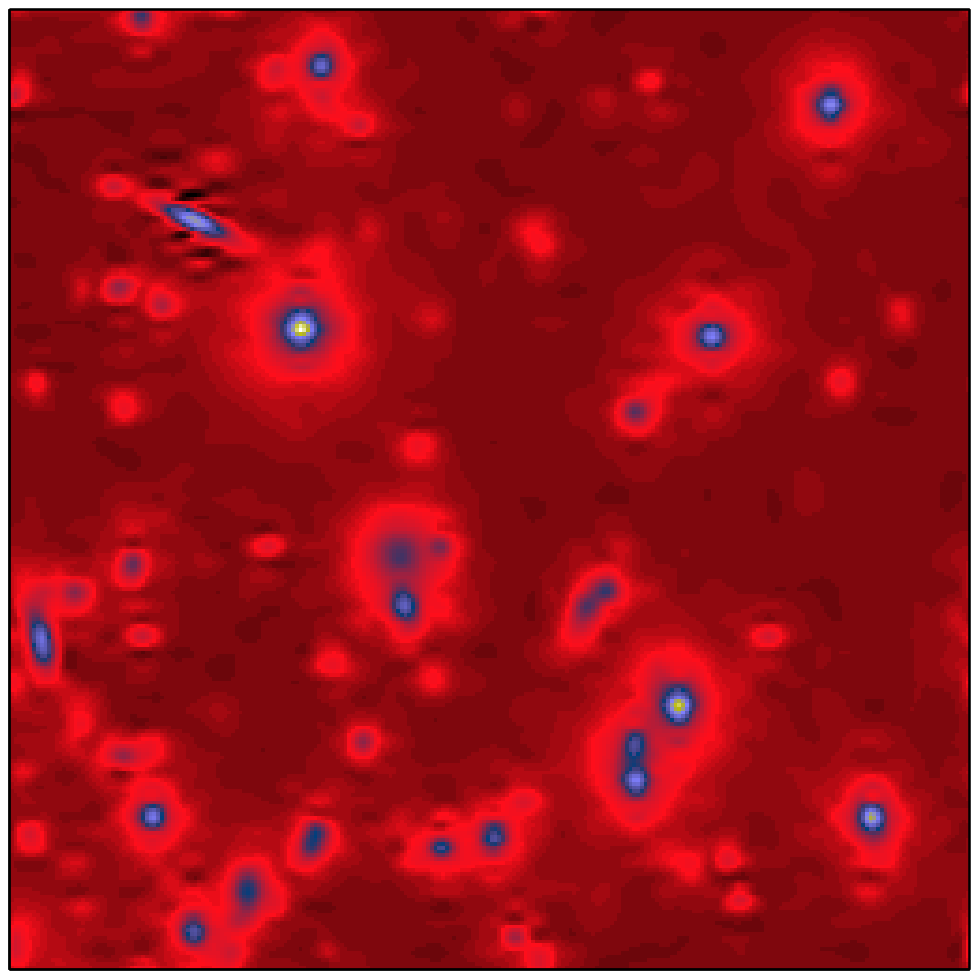} \\
   (g) & (h) & (i) \\
    \includegraphics[width=0.33\linewidth]{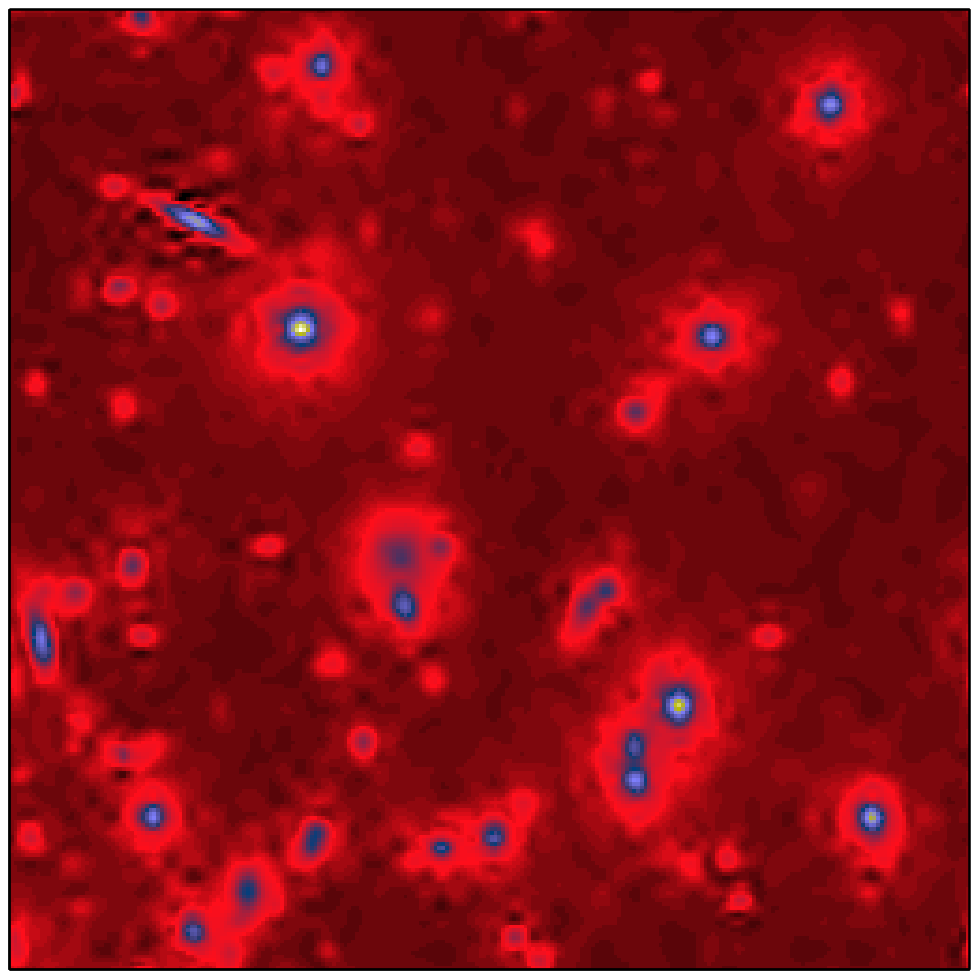}& 
    \includegraphics[width=0.33\linewidth]{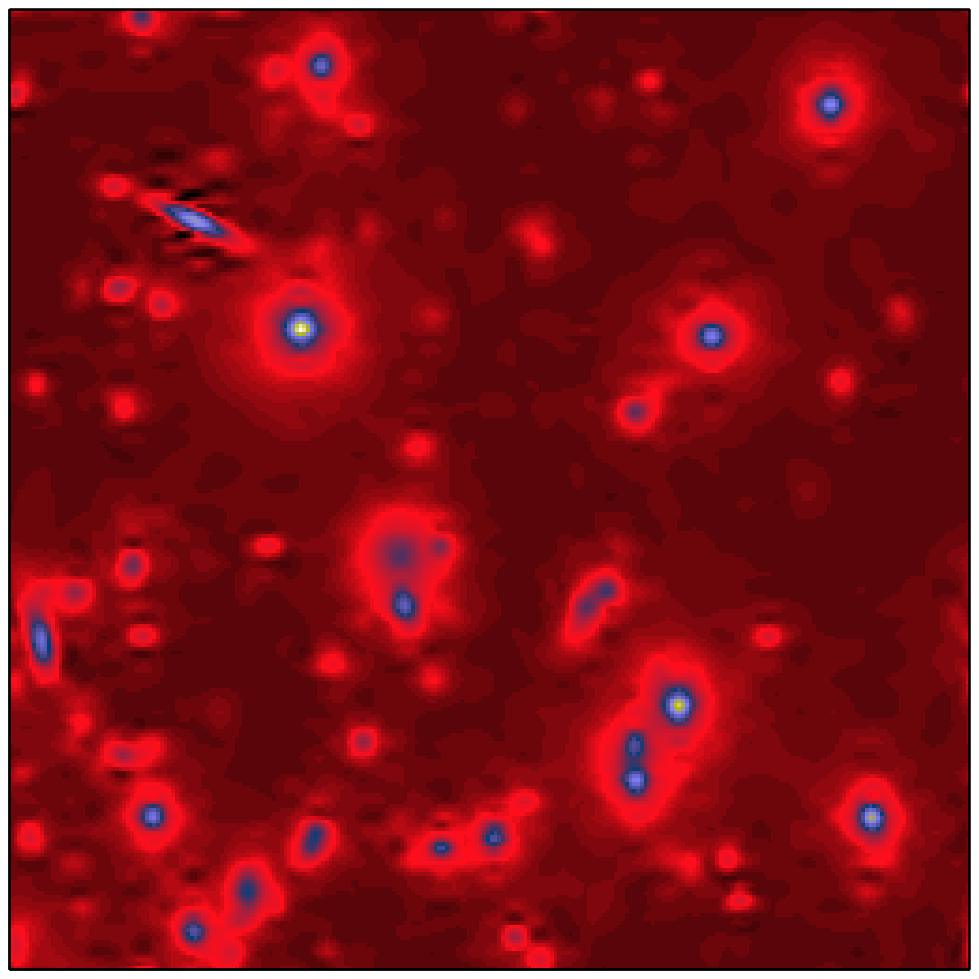} &
    \includegraphics[width=0.33\linewidth]{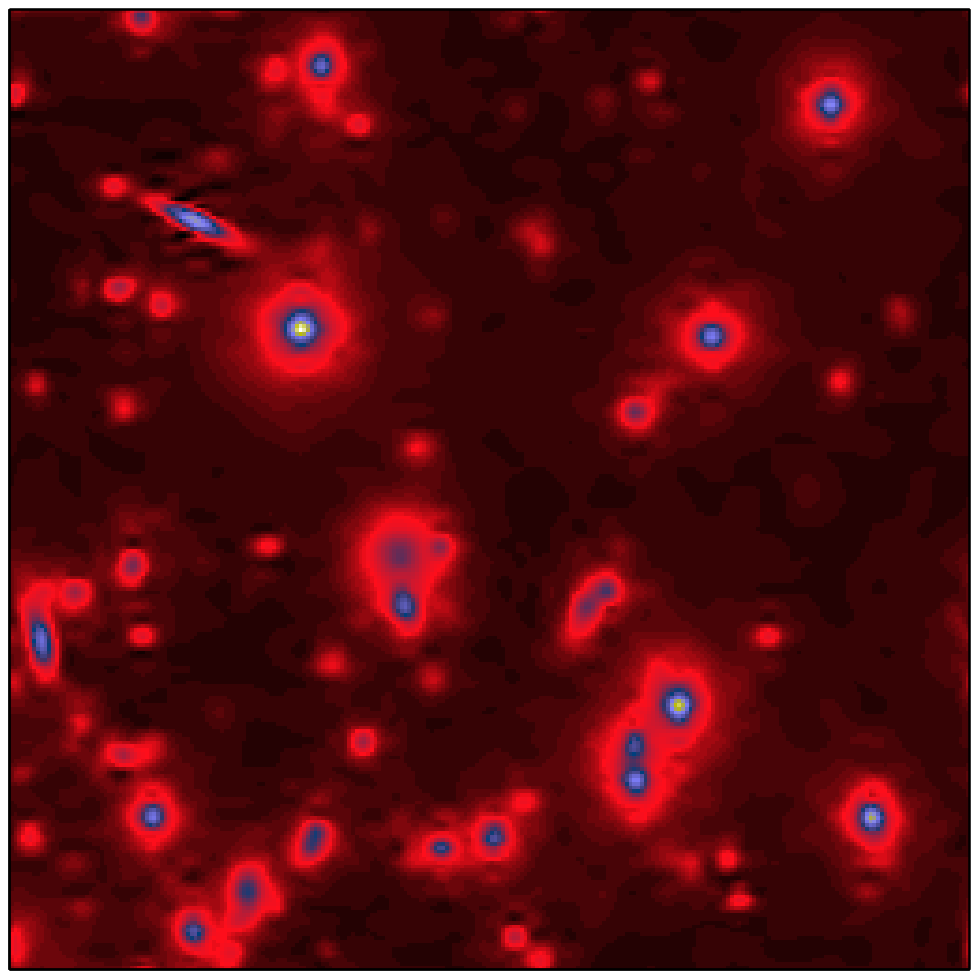}  \\
    (j) & (k) & (l)  \\
  \end{tabular}
  \caption{Deconvolution of the sky simulation. (a) Original, (b) Blurred, (c) Blurred
     and noisy, (d) RL-MRS \cite{Starck2006}, (e) RL-TV \cite{Dey2004}, (f) FTITPR \cite{Willett2004},
    (g) NaiveG \cite{Vonesch2007}, (h) StabG \cite{Dupe2009c}, (i) PIDAL-FA \cite{Figueiredo2010},
    (j) PIDAL-FS \cite{Figueiredo2010}, (k) Prox-FA, (l) Prox-FS.}
 \label{fig:sky}
\end{figure}

\subsection{Intensity level influence}
\label{sec:imaging-dataset}

To assess the impact of the intensity level (average photon counts), the different algorithms in the comparative study
were applied to an image of the planet Saturn at several intensity levels. In this experiment, the PSF was a $7 \times 7$
moving average. The regularization parameter $\eqi$ was chosen using the GCV criterion \eqref{eq:6}. Given that this image is piecewise smooth away from smooth curvilinear structures (e.g. the rings),
the curvelet dictionary will be a very good candidate to sparsify it \cite{CandesDonohoCurvelets}.



\begin{table}[htp]
  \centering
  \footnotesize
  \begin{tabular}{|c||c|c|c|c|c|c|}
    \hline Intensity  Level & NaiveG & StabG & PIDAL-FA & PIDAL-FS & Prox-FA & Prox-FS \\
    \hline \hline 5 & 0.70 (33.49) & 1.14 (54.54)  & 0.34 (16.04)   & 0.24 (11.78)   & 0.42 (20.09)    & 0.25 (11.96) \\
    \hline 30         & 1.27 (10.13) & 6.62 (52.78)   & 4.22 (33.62)   & 1.41 (11.23)   & 2.35 (18.74)   & 2.09 (16.66) \\
    \hline 100       & 3.24 (7.75)   & 26.73 (63.94) & 23.40 (55.97) & 12.46 (29.81) & 7.88 (18.85)   & 7.31 (17.49) \\
    \hline 255       & 7.48 (7.02)   & 80.45 (75.46) & 73.83 (69.26) & 49.36 (46.30) & 21.02 (19.72) & 18.78 (17.62) \\
    \hline
  \end{tabular}
  \caption{Average MAE values and relative MAE (in parentheses and in percent) for the Saturn image as a function of the intensity
    level with the naive Gaussian (NaiveG)~\cite{Vonesch2008}, stabilized gaussian (StabG)~\cite{Dupe2009c},
    PIDAL-FA and PIDAL-FS~\cite{Figueiredo2010}, and our two algorithms Prox-FA and Prox-FS.}
  \label{tab:maeres}
\end{table}

At each intensity value, ten noisy and blurred replications were generated and the MAE was computed for each
deconvolution algorithm. The average and relative (i.e. divided by the mean intensity) MAE values over the ten
realizations are summarized in Table~\ref{tab:maeres}. As expected, for low intensity levels, the best results are
obtained for the methods using data fidelity terms that account properly for Poisson noise. For higher levels, the
performance of NaiveG gets better as the Poisson distribution tends to normal with increasing intensity. The difference
in performance, especially for the high intensity regime, between our approaches and the PIDAL ones seem to originate
from the optimization algorithms used (and their convergence then). 

For this image, it seems that the synthesis prior is slightly better than the analysis one. This might seem contradictory with results reported in the Gaussian noise case by a few authors where the analysis prior was observed to be superior to its synthesis counterpart \cite{EladAnalysisSynthesis07,CandesIRL108,SelesnickSPIE09}. But one should avoid to infer strong conclusions from such limited body of experimental work. For instance, the noise is different in our setting, and the deep phenomena underlying differences and similarities between analysis and synthesis priors are still not very well understood and more investigation is needed in this direction both on the theoretical and practical sides.


\begin{figure}[htp]
  \centering
  \begin{tabular}{c@{ }c@{ }c}
    \includegraphics[width=0.33\linewidth]{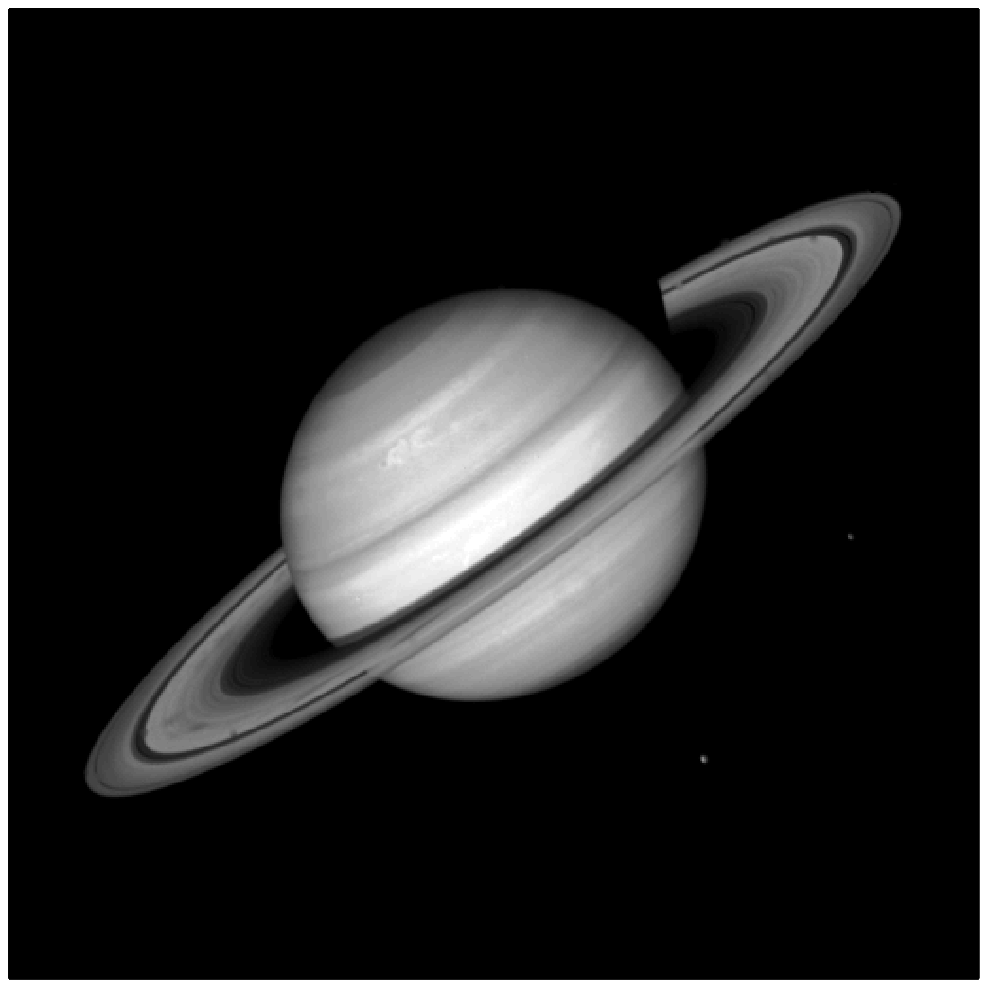} &
    \includegraphics[width=0.33\linewidth]{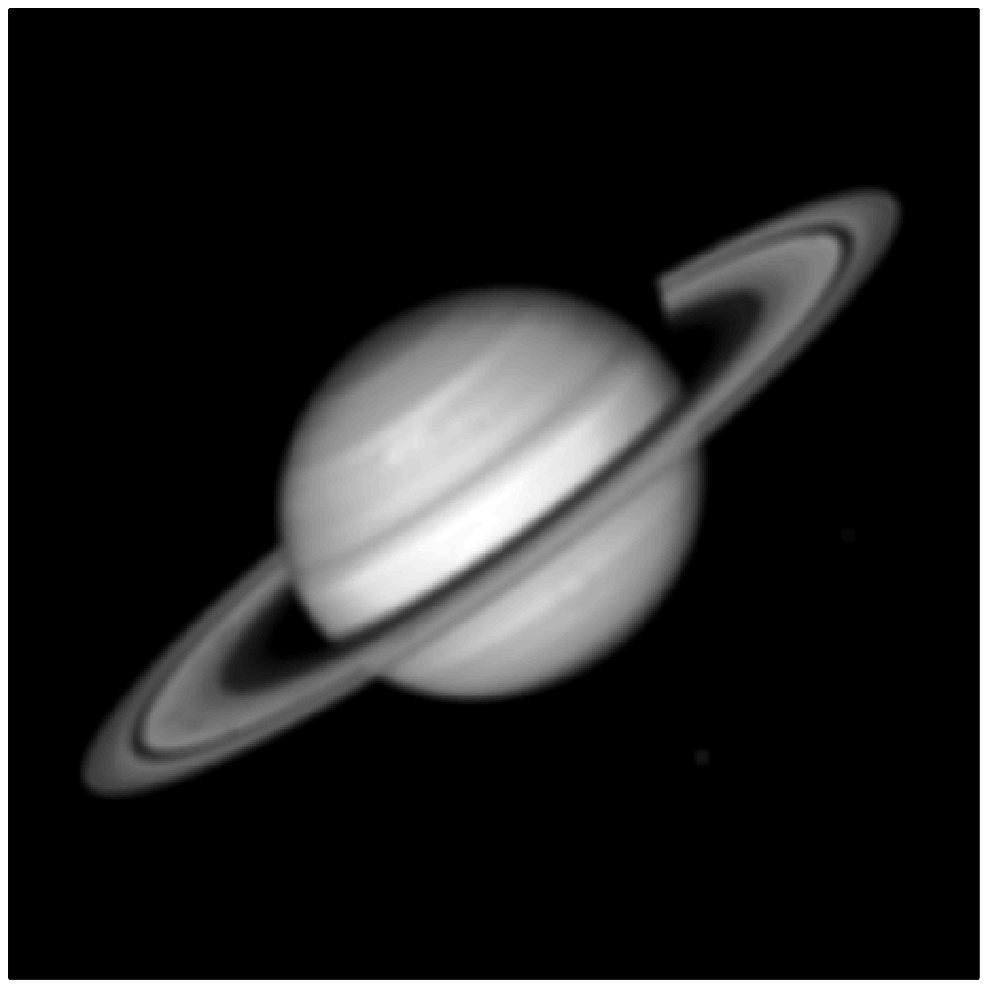} &
    \includegraphics[width=0.33\linewidth]{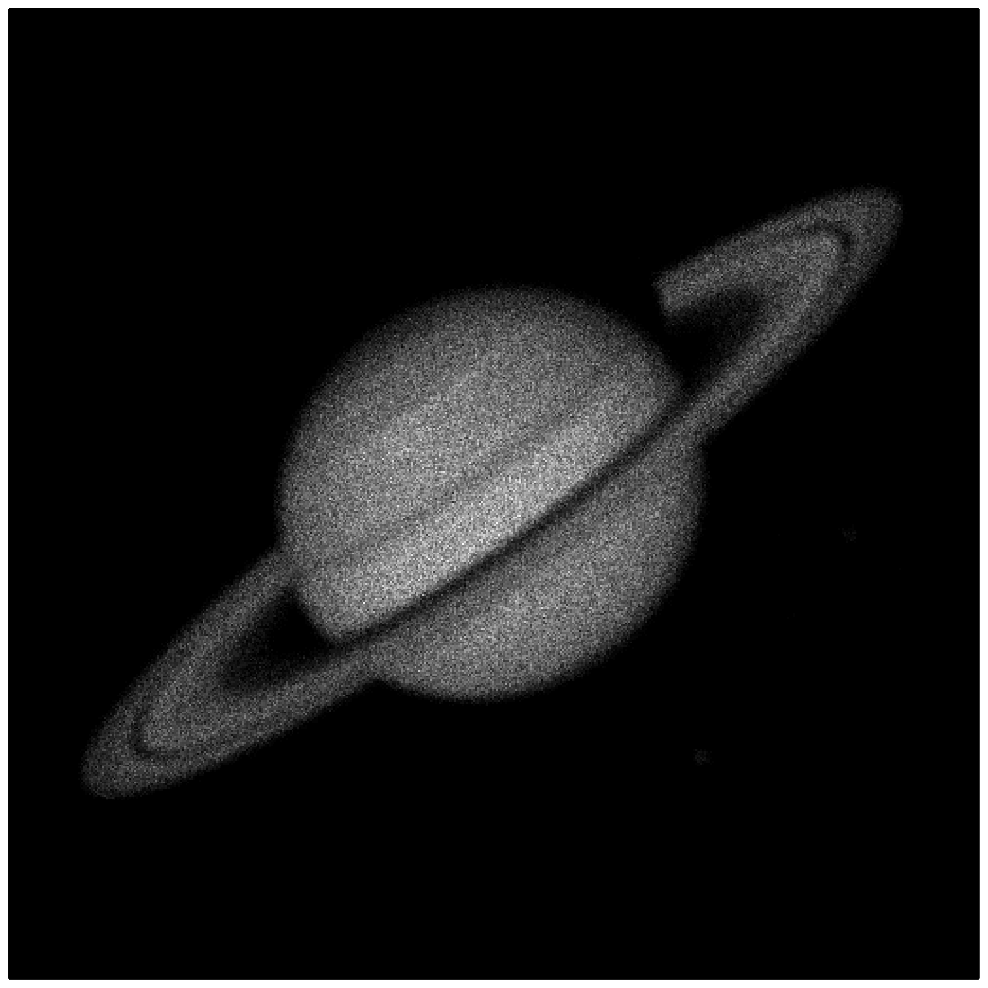} \\
    (a) & (b) & (c) \\
    \includegraphics[width=0.33\linewidth]{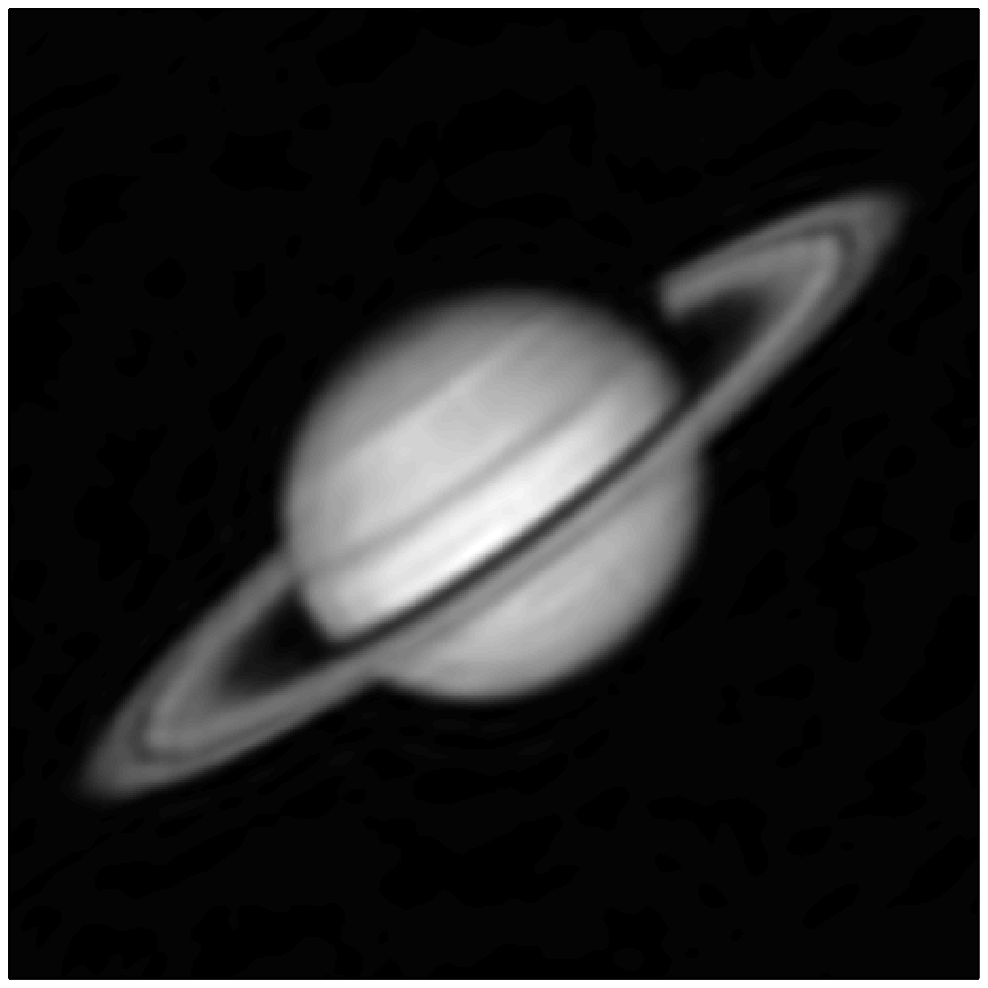} &
    \includegraphics[width=0.33\linewidth]{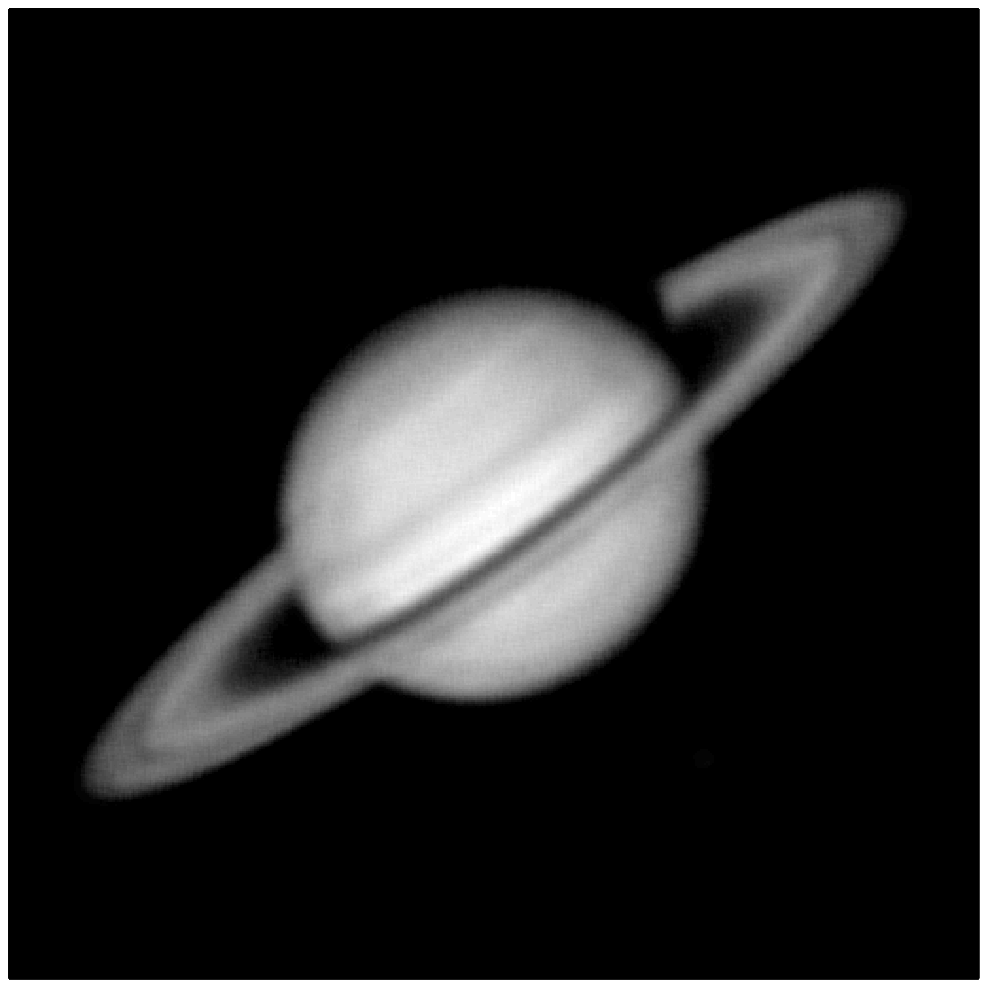} &
    \includegraphics[width=0.33\linewidth]{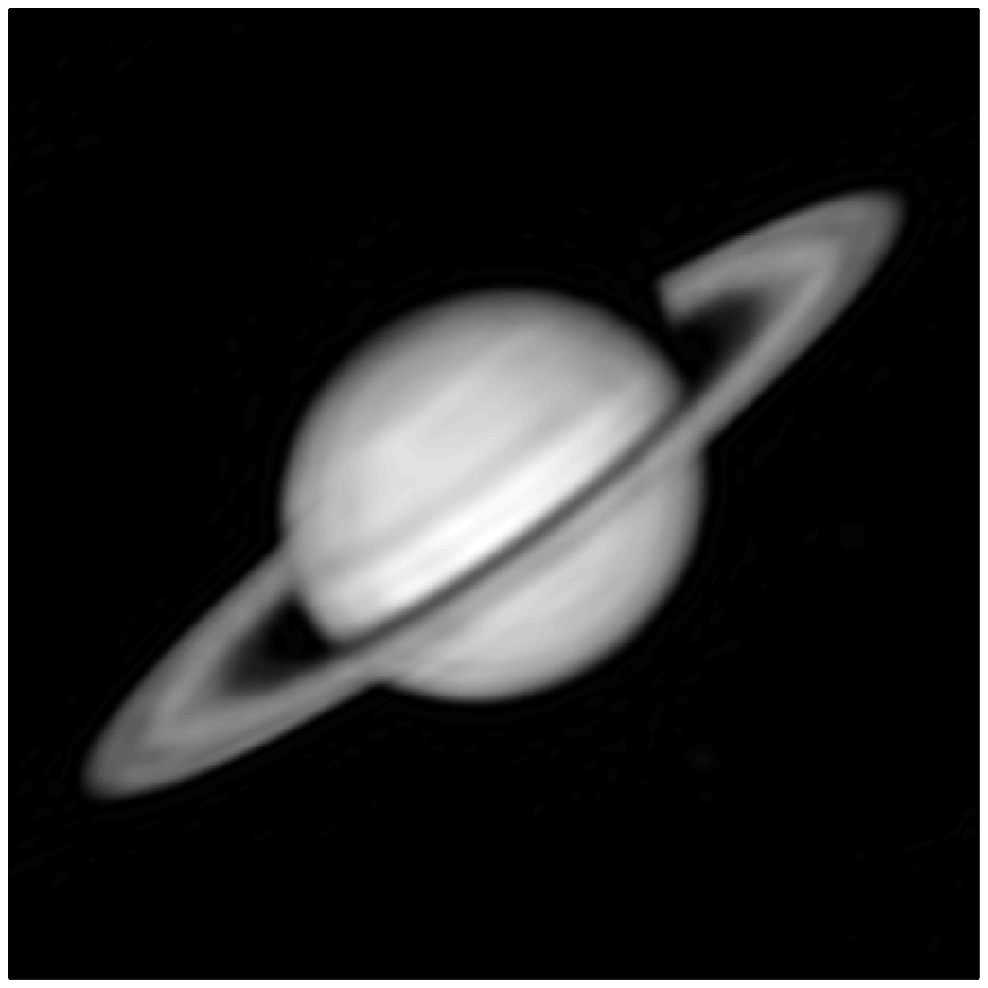} \\
    (d) & (e) & (f) \\
    \includegraphics[width=0.33\linewidth]{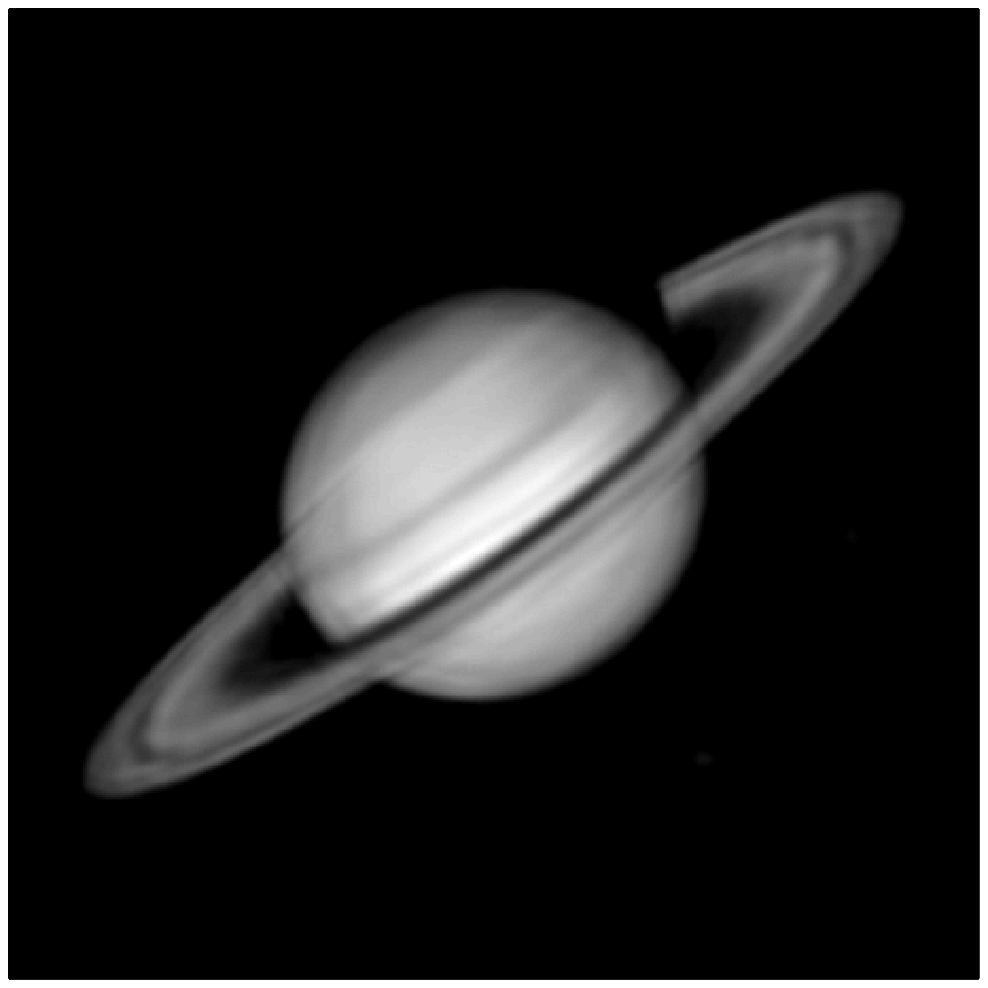} &
    \includegraphics[width=0.33\linewidth]{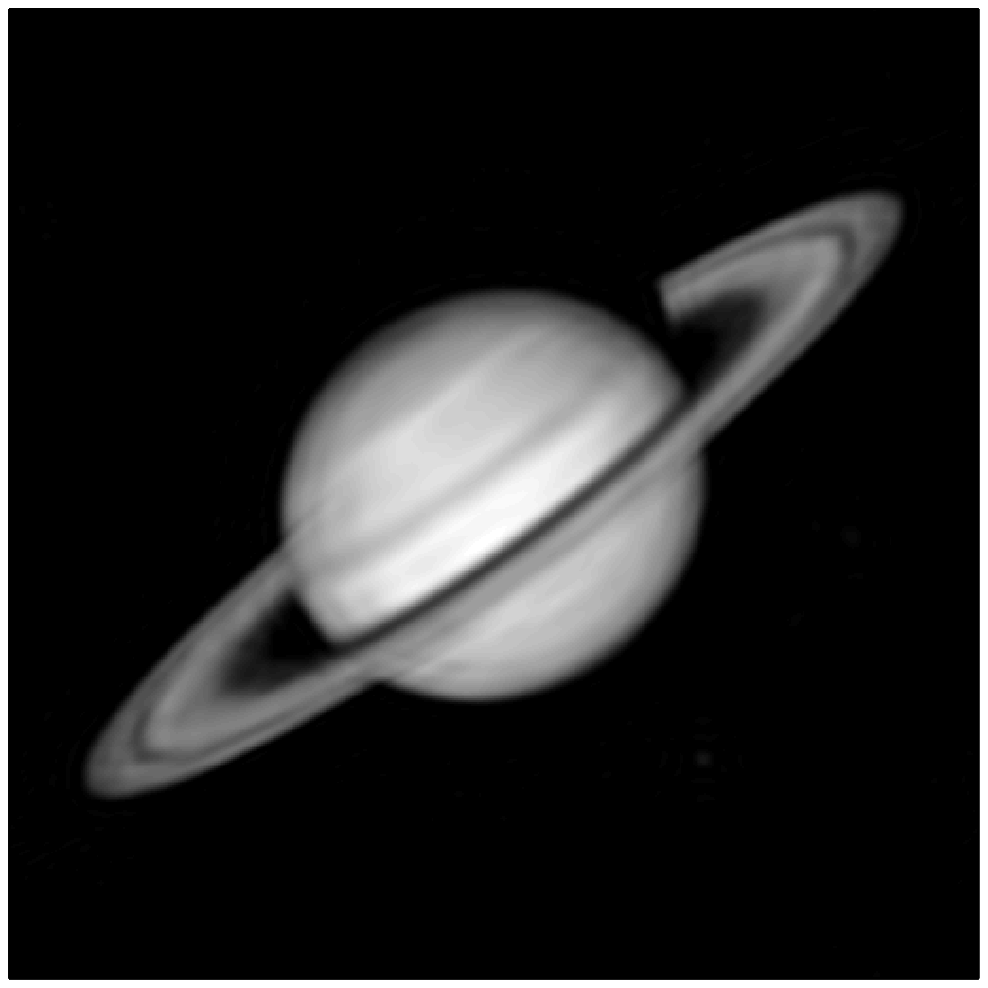} &
    \includegraphics[width=0.33\linewidth]{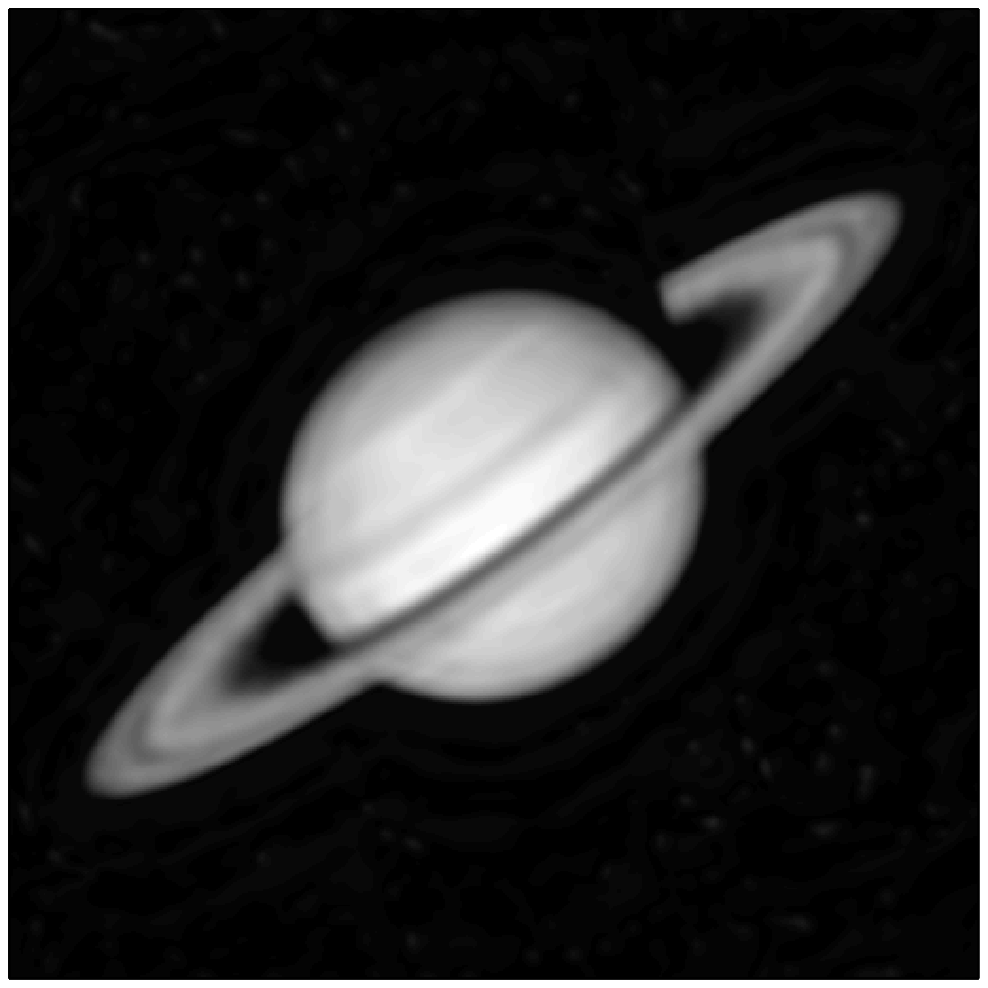} \\
    (g) & (h) & (i)
 \end{tabular}
  \caption{Deconvolution of Saturn image. (a) original, (b) Convolved, (c) Convolved and Noisy, (d) NaiveG \cite{Vonesch2008}, 
  (e) StabG \cite{Dupe2009c}, (f) PIDAL-FA \cite{Figueiredo2010}, (g) PIDAL-FS \cite{Figueiredo2010}, (h) Prox-FA, (i) Prox-FS.}
  \label{fig:saturn}
\end{figure}

The visual results obtained for a maximal intensity of 30 are portrayed in Figure~\ref{fig:saturn}. For all the methods,
the curvelet transform is able to capture most of the curved structures inside the image and then leads to a good
restoration of them. The original image contains some tiny satellites at the bottom right that appear as localized
spikes. Prox-FA was able to recover one of them, though it is hardly visible in the restored image. Prox-FS however
recovers several structures that look like background artifacts including the satellites. A wise solution would be to
use another transform in the dictionary beside curvelets, which is well-adapted to point-like structures (e.g. wavelets
or even Diracs).


\section{Conclusion}
\label{sec:conclusion}

In this paper, a novel sparsity-based fast iterative thresholding deconvolution algorithm that takes account of the
presence of Poisson noise was presented. The Poisson noise was handled properly using its associated likelihood function
to construct the data fidelity term. A careful theoretical study of the optimization problems and convergence guarantees of the iterative algorithms were provided. Several experimental experiments have shown the capabilities of our approach, which compares favorably with some state-of-the-art algorithms.

The present work may be extended along several lines. For example, our approach uses a proximal algorithm which is a
generalization of the Douglas-Rachford splitting applied to the primal problem, but to the best of our knowledge, its
convergence rate (even on the objective) is not known. Furthermore, in this paper, the regularization parameter was chosen based on the GCV formula proposed, see \eqref{eq:6}. But this formula was derived using some heuristics and was valid for the stabilized version of the deconvolution problem. This objective and automatic selection of the regularization parameter remains an open issue that is worth investigating in the future.

\paragraph{Acknowledgment}
This work has been supported by the European Research Council grant SparseAstro (ERC-228261).

\appendix
\section{Proof of Theorem~\ref{th:compo}}

As $\Amb$ is an affine operator, it preserves convexity. Thus, $f \circ \Amb$ is convex. Moreover, $\Amb$ is continuous and $f$ is lsc, and so is $f \circ \Amb$. By the domain qualification condition, we then have $f \circ \Amb \in \Gamma_0(\RR^m)$. Note that the domain assumption in the theorem is not needed for frames by surjectivity of $\Fmb$. 

By Fenchel-Rockafellar duality \cite[Section III.4]{Ekeland74}, we have
\begin{eqnarray}
\label{eq:proxprimalframe}
\bar{p} = \prox_{\prox_{f\circ\Amb}} &\iff& \bar{p} = \arginf_{p \in \RR^n} \frac{1}{2}\norm{p - x}^2 + f \circ \Amb (p) \\
\label{eq:proxdualframe}
			     	   &\iff& \bar{u} = \argmin_{u \in \RR^m} \frac{1}{2}\norm{x - \Fmb\trans u}^2 + \pds{u}{y} + f^*(u) ~,
\end{eqnarray}
and the primal solution $\bar{p}$ is recovered from the dual solution $\bar{u}$ as
\[
\bar{p} = x - \Fmb\trans \bar{u} ~.
\]
\begin{enumerate}[(i)]
\item $\Fmb$ is a tight frame with constant $c > 0$. Applying the minimality condition to \eqref{eq:proxdualframe}, using the fact that $\Fmb\Fmb\trans=c\Id$ and Moreau Identity, we obtain
\begin{eqnarray}
(\Fmb x - y) - c\bar{u} &\in& \partial f^*(\bar{u}) \nonumber \\
c^{-1}(\Fmb x - y) - \bar{u} &\in& \partial(c^{-1}f^*)(\bar{u}) \nonumber \\
\bar{u} &=& \prox_{c^{-1}f^*}\parenth{c^{-1}(\Fmb x - y)} \nonumber \\
\bar{u} &=& c^{-1}\parenth{\Id - \prox_{c f}}(\Fmb x - y) ~,
\end{eqnarray}
which leads to \eqref{eq:86}.

\item $\Fmb$ is a general frame operator. From \eqref{eq:proxdualframe}, $\frac{1}{2}\norm{x - \Fmb\trans .}^2 + \pds{.}{y}$ is continuous with $c_2$-Lipschitz continuous gradient. Therefore, applying the forward-backward splitting to the minimization problem \eqref{eq:proxdualframe}, we obtain \eqref{eq:1}. From strong convexity of the dual problem \eqref{eq:proxdualframe}, the dual minimizer is unique. 

To get the linear convergence rate, the proof is very classical and is an extension of the projected gradient descent one that can be found in classical monographs such as \cite[Theorem~8.6.2]{Ciarlet}.

\item In all other cases, the forward-backward splitting applied to \eqref{eq:proxdualframe} converges provided that $0 < \inf_t \tau_t \leq \sup_t \tau_t < 2/\opnorm{\Fmb}^2=2/c_2$. The proof of the convergence rate is technical and follows the same lines as \cite[Theorem~1]{Fadili2009b}.

\end{enumerate}

\bibliographystyle{elsarticle-harv}
\bibliography{paperbib}

\begin{thebibliography}{44}
\expandafter\ifx\csname natexlab\endcsname\relax\def\natexlab#1{#1}\fi
\expandafter\ifx\csname url\endcsname\relax
  \def\url#1{\texttt{#1}}\fi
\expandafter\ifx\csname urlprefix\endcsname\relax\def\urlprefix{URL }\fi

\bibitem[{Andrews and Hunt(1977)}]{AndrewsHunt77}
Andrews, H.~C., Hunt, B.~R., 1977. Digital Image Restoration. Prentice-Hall.

\bibitem[{Anscombe(1948)}]{Anscombe48}
Anscombe, F.~J., 1948. {T}he {T}ransformation of {P}oisson, {B}inomial and
  {N}egative-{B}inomial {D}ata. Biometrika 35, 246--254.

\bibitem[{Bertero et~al.(2009)Bertero, Boccacci, Desider{\`a}, and
  Vicidomini}]{Bertero2009}
Bertero, M., Boccacci, P., Desider{\`a}, G., Vicidomini, G., 2009. {Image
  deblurring with Poisson data: from cells to galaxies}. Inverse Problems
  25~(123006).

\bibitem[{Cand\`es and Donoho(1999)}]{CandesDonohoCurvelets}
Cand\`es, E.~J., Donoho, D.~L., 1999. Curvelets -- a surprisingly effective
  nonadaptive representation for objects with edges. In: Cohen, A., Rabut, C.,
  Schumaker, L. (Eds.), Curve and Surface Fitting: Saint-Malo 1999. Vanderbilt
  University Press, Nashville, TN.

\bibitem[{Cand\`es et~al.(2008)Cand\`es, Wakin, and Boyd}]{CandesIRL108}
Cand\`es, E.~J., Wakin, M.~B., Boyd, S.~P., 2008. Enhancing sparsity by
  reweighted {L1} minimization. Journal of Fourier Analysis and Applications
  14~(5), 877--905.

\bibitem[{Ciarlet(1982)}]{Ciarlet}
Ciarlet, P., 1982. Introduction \`a l'Analyse Num{\'e}rique Matricielle et \`a
  l'Optimisation. Paris.

\bibitem[{Combettes and Pesquet(2007{\natexlab{a}})}]{Combettes2007a}
Combettes, P.~L., Pesquet, J.-., 2007{\natexlab{a}}. A {D}ouglas-{R}achford
  splittting approach to nonsmooth convex variational signal recovery. IEEE
  Journal of Selected Topics in Signal Processing 1~(4), 564--574.

\bibitem[{Combettes and Pesquet(2007{\natexlab{b}})}]{Combettes2007}
Combettes, P.~L., Pesquet, J.-C., November 2007{\natexlab{b}}. Proximal
  thresholding algorithm for minimization over orthonormal bases. SIAM Journal
  on Optimization 18~(4), 1351--1376.

\bibitem[{Combettes and Pesquet(2008)}]{Combettes2008}
Combettes, P.~L., Pesquet, J.-C., 2008. A proximal decomposition method for
  solving convex variational inverse problems. Inverse Problems 24~(6).

\bibitem[{Combettes and Wajs(2005)}]{Combettes2005}
Combettes, P.~L., Wajs, V.~R., 2005. Signal recovery by proximal
  forward-backward splitting. SIAM Multiscale Model. Simul. 4~(4), 1168--1200.

\bibitem[{Daubechies et~al.(2004)Daubechies, Defrise, and Mol}]{Daubechies2004}
Daubechies, I., Defrise, M., Mol, C.~D., 2004. An iterative thresholding
  algorithm for linear inverse problems with a sparsity constraints. Comm. Pure
  Appl. Math. 112, 1413--1541.

\bibitem[{Dey et~al.(2004)Dey, Blanc-F{\'e}raud, Zerubia, Zimmer, Olivo-Marin,
  and Kam}]{Dey2004}
Dey, N., Blanc-F{\'e}raud, L., Zerubia, J., Zimmer, C., Olivo-Marin, J.-C.,
  Kam, Z., 2004. A deconvolution method for confocal microscopy with total
  variation regularization. In: Proceedings of the 2004 IEEE International
  Symposium on Biomedical Imaging: From Nano to Macro, Arlington, VA, USA,
  15-18 April 2004. IEEE, pp. 1223--1226.

\bibitem[{Dup{\'e} et~al.(2009)Dup{\'e}, Fadili, and Starck}]{Dupe2009c}
Dup{\'e}, F.-X., Fadili, M.-J., Starck, J.-L., 2009. A proximal iteration for
  deconvolving poisson noisy images using sparse representations. IEEE
  Transactions on Image Processing 18~(2), 310--321.

\bibitem[{E.~J.~Cand{\`e}s and Randall(2011)}]{CandesAnalysis}
E.~J.~Cand{\`e}s, Y.~Eldar, D.~N., Randall, P., 2011. Compressed sensing with
  coherent and redundant dictionaries. Applied and Computational Harmonic
  AnalysisTo appear.

\bibitem[{Eckstein and Bertsekas(1992)}]{Eckstein92}
Eckstein, J., Bertsekas, D.~P., 1992. On the {D}ouglas-{R}achford splitting
  method and the proximal point algorithm for maximal monotone operators. Math.
  Programming 55, 293--318.

\bibitem[{Ekeland and Temam(1974)}]{Ekeland74}
Ekeland, I., Temam, R., 1974. Analyse convexe et probl\`emes variationnels.
  Dunod.

\bibitem[{Elad et~al.(2007)Elad, Milanfar, and
  Rubinstein}]{EladAnalysisSynthesis07}
Elad, M., Milanfar, P., Rubinstein, R., 2007. Analysis versus synthesis in
  signal priors. Inverse Problems 23~(3), 947--968.

\bibitem[{Fadili and Peyr{\'e}(2011)}]{Fadili2009b}
Fadili, M., Peyr{\'e}, G., 2011. Total variation projection with first order
  schemes. IEEE Transactions on Image Processing 20~(3), 657--669.

\bibitem[{Fadili et~al.(2007)Fadili, Starck, and Murtagh}]{Fadili2006}
Fadili, M., Starck, J.-L., Murtagh, F., 2007. Inpainting and zooming using
  sparse representations. The Computer Journal 52~(1), 64--79.

\bibitem[{Fadili and Starck(2006)}]{Fadili2006a}
Fadili, M.~J., Starck, J.-L., 2006. Sparse representation-based image
  deconvolution by iterative thresholding. In: ADA IV. Elsevier, France.

\bibitem[{Figueiredo and Bioucas-Dias(2010)}]{Figueiredo2010}
Figueiredo, M., Bioucas-Dias, J., 2010. Restoration of poissonian images using
  alternating direction optimization. IEEE Transactions on Image Processing,{
  }Submitted.
\newline\urlprefix\url{\url{http://arxiv.org/abs/1001.2244}}

\bibitem[{Figueiredo and Nowak(2003)}]{Figueiredo2003}
Figueiredo, M., Nowak, R., 2003. An em algorithm for wavelet-based image
  restoration. IEEE Transactions on Image Processing 12, 906--916.

\bibitem[{Gabay(1983)}]{Gabay83}
Gabay, D., 1983. Applications of the method of multipliers to variational
  inequalities. In: Fortin, M., Glowinski, R. (Eds.), Augmented Lagrangian
  Methods: Applications to the Solution of Boundary-Value Problems.
  North-Holland, Amsterdam.

\bibitem[{Jammal and Bijaoui(2004)}]{Bijaoui04}
Jammal, G., Bijaoui, A., 2004. Dequant: a flexible multiresolution restoration
  framework. Signal Processing 84, 1049--1069.

\bibitem[{Jansson(1997)}]{Jansson97}
Jansson, P., 1997. Deconvolution of Images and Spectra. Academic Press, New
  York.

\bibitem[{Lemar\'echal and Hiriart-Urruty(1996)}]{LemarechalHiriart96}
Lemar\'echal, C., Hiriart-Urruty, J.-B., 1996. Convex Analysis and Minimization
  Algorithms I, 2nd Edition. Springer.

\bibitem[{Mallat(1998)}]{Mallat98}
Mallat, S.~G., 1998. A Wavelet Tour of Signal Processing, 2nd Edition. Academic
  Press.

\bibitem[{Moreau(1962)}]{Moreau1962}
Moreau, J.-J., 1962. {Fonctions convexes duales et points proximaux dans un
  espace hilbertien}. CRAS S\'er. A Math. 255, 2897--2899.

\bibitem[{Moreau(1963)}]{Moreau1963}
Moreau, J.-J., 1963. {Propri\'et\'es des applications ``prox''}. Comptes Rendus
  de l'Acad\'emie des Sciences S\'erie A Math\'ematiques 256, 1069--1071.

\bibitem[{Moreau(1965)}]{Moreau1965}
Moreau, J.-J., 1965. {Proximit\'e et dualit\'e dans un espace hilbertien}.
  Bulletin de la Soci\'et\'e Math\'ematique de France 93, 273--299.

\bibitem[{Pawley(2005)}]{Pawley2005}
Pawley, J., 2005. Handbook of Confocal Microscopy. Plenum Press.

\bibitem[{Pustelnik et~al.(2011)Pustelnik, Chaux, and Pesquet}]{Pustelnik2009}
Pustelnik, N., Chaux, C., Pesquet, J.-C., 2011. Parallel proximal algorithm for
  image restoration using hybrid regularization. IEEE Trans. on image
  processingTo appear.

\bibitem[{Sarder and Nehorai(2006)}]{Sarder2006}
Sarder, P., Nehorai, A., 2006. {Deconvolution Method for {3-D} Fluorescence
  Microscopy Images}. IEEE Sig. Pro. Mag. 23, 32--45.

\bibitem[{Selesnick and Figueiredo(2009)}]{SelesnickSPIE09}
Selesnick, I., Figueiredo, M., 2009. Signal restoration with overcomplete
  wavelet transforms: Comparison of analysis and synthesis priors. In: SPIE
  conference on Signal and Image Processing: Wavelet Applications in Signal and
  Image Processing XIII.

\bibitem[{Spingarn(1983)}]{Spingarn}
Spingarn, J.~E., 1983. Partial inverse of a monotone operator. Appl. Math.
  Optim. 10~(3), 247--265.

\bibitem[{Starck et~al.(1995)Starck, Bijaoui, and Murtagh}]{Starck95}
Starck, J.-L., Bijaoui, A., Murtagh, F., 1995. Multiresolution support applied
  to image filtering and deconvolution. CVGIP: Graphical Models and Image
  Processing 57, 420--431.

\bibitem[{Starck and Murtagh(1994)}]{Starck94}
Starck, J.-L., Murtagh, F., 1994. Image restoration with noise suppression
  using the wavelet transform. Astronomy and Astrophysics 288, 343--348.

\bibitem[{Starck and Murtagh(2006)}]{Starck2006}
Starck, J.-L., Murtagh, F., 2006. Astronomical Image and Data Analysis.
  Springer.

\bibitem[{Starck et~al.(2010)Starck, Murtagh, and Fadili}]{FadiliStarckBook09}
Starck, J.-L., Murtagh, F., Fadili, M., 2010. Sparse Signal and Image
  Processing: Wavelets, Curvelets and Morphological Diversity. Cambridge
  University Press, Cambridge, UK, in press.

\bibitem[{Starck et~al.(2003)Starck, Nguyen, and Murtagh}]{Starck03}
Starck, J.-L., Nguyen, M., Murtagh, F., 2003. Wavelets and curvelets for image
  deconvolution: a combined approach. Signal Processing 83, 2279--2283.

\bibitem[{Stark(1987)}]{Stark1987}
Stark, H., 1987. {Image Recovery: Theory and Application}. Elsevier Science \&
  Technology Books.

\bibitem[{Vonesch and Unser(2007)}]{Vonesch2007}
Vonesch, C., Unser, M., 2007. Fast wavelet-regularized image deconvolution. In:
  Proceedings of the 2007 IEEE International Symposium on Biomedical Imaging:
  From Nano to Macro, Washington, DC, USA, April 12-16, 2007. pp. 608--611.

\bibitem[{Vonesch and Unser(2008)}]{Vonesch2008}
Vonesch, C., Unser, M., 2008. {A fast thresholded Landweber algorithm for
  general wavelet bases: Application to 3D deconvolution microscopy}. In:
  Proceedings of the 2008 IEEE International Symposium on Biomedical Imaging:
  From Nano to Macro, Paris, France, May 14-17, 2008. pp. 1351--1354.

\bibitem[{Willett and Nowak(2004)}]{Willett2004}
Willett, R., Nowak, R., 2004. Fast multiresolution photon-limited image
  reconstruction. In: IEEE ISBI.

\end{thebibliography}

\end{document}